\documentclass{amsart}
\date{7.7.20}
\usepackage{bm} 

\newtheorem{theorem}{Theorem}
\newtheorem{corollary}{Corollary}

\newtheorem{lemma}{Lemma}
\newtheorem{proposition}{Proposition}

\def\rd{\mathrm{d}}
\def\ri{\mathrm{i}}

\def\cz{\mathbb{C}} 
\def\nz{\mathbb{N}}
\def\gz{\mathbb{Z}}
\def\zp{\dot{\gz}}
\def\rz{\mathbb{R}}

\def\bbS{\mathbb{S}}

\def\cD{\mathcal{D}} 
\def\cE{\mathcal{E}}

\def\cK{\mathcal{K}}
\def\cS{\mathcal{S}}

\def\gH{\mathfrak{H}} 
\def\gh{\mathfrak{h}}

\def\ghpk{{\gh^+_\kappa}}
\def\ppp{\Pi^+_\kappa}
\def\gQ{\mathfrak{Q}}
\def\gS{\mathfrak{S}}

\def\dr{\mathrm{d}r}
\def\dx{\mathrm{d}x}
\def\dy{\mathrm{d}y}
\def\domega{\rd\omega}

\newcommand{\one}{\mathbf{1}}

\def\const{\mathrm{const}\,}

\def\balpha{\bm{\alpha}}
\def\bsigma{\bm{\sigma}}
\def\bp{\bm{p}}

\DeclareMathOperator{\sgn}{sgn}
\DeclareMathOperator{\tr}{tr}
 
\begin{document}
\title[Strong Scott Conjecture]{Proof of the Strong Scott Conjecture
  for Heavy Atoms: the Furry Picture}

\author[K. Merz]{Konstantin Merz} \address{Institut f\"ur Analysis und
  Algebra\\ Carolo-Wilhelmina\\ Universit\"atsplatz 2\\ 38106
  Braunschweig\\ Germany} \email{k.merz@tu-bs.de}

\author[H. Siedentop]{Heinz Siedentop}
\address{Mathematisches Institut\\
  Ludwig-Maximilians-Universit\"at M\"unchen\\ Theresienstra\ss e 39\\
  80333 M\"unchen\\Germany} \email{h.s@lmu.de}

\begin{abstract}
  We prove the convergence of the density on the scale $Z^{-1}$ to the
  density of the Bohr atom (with infinitely many electrons) (strong
  Scott conjecture) for a model that is known to describe heavy atoms
  accurately.
\end{abstract}

\maketitle
\section{Introduction}
The quest for ground state properties of Coulomb systems like atoms,
molecules, and solids is one of the central topics in physics and
chemistry. However, it became clear right after the discovery of
quantum mechanics (Heisenberg \cite{Heisenberg1925}) that -- not much
different from classical mechanics -- one-particle problems like the
hydrogen atom can be solved analytically (Pauli \cite{Pauli1926}) but
problems with several electrons need suitable approximations. Within
two years after the advent of quantum mechanics Thomas
\cite{Thomas1927} and Fermi \cite{Fermi1927,Fermi1928}) developed an
approximation, now called Thomas-Fermi theory -- for predicting the
ground state energies and densities of large atoms. About fifty years
later Lieb and Simon \cite{LiebSimon1977} showed in their seminal work
that, indeed, the asymptotic behavior of atomic energies for large
atomic numbers $Z$ is given by the Thomas-Fermi energy, namely
$e_\mathrm{TF}Z^{7/3}$ and that the suitably renormalized ground state
density of large atoms on the scale $Z^{-1/3}$ converges to the
hydrogenic Thomas-Fermi density.

The Thomas-Fermi theory is the simplest example of what is called
density functional theory. However, already the next order correction
is not easily connected with the first correction of Thomas-Fermi
theory, the Thomas-Fermi-Weizs\"acker theory. It requires a
renormalization of the constant in front of the inhomogeneity
correction (Yonei and Tomishima \cite{YoneiTomishima1965}, see also
Lieb \cite{Lieb1982A,LiebLiberman1982}). In fact the next order energy
correction was predicted by Scott \cite{Scott1952} as stemming
entirely from the electrons on the scale $Z^{-1}$ where the
interaction between the electrons is completely dominated by the
electron-nucleus interaction. He suggested that the correction is the
same as for non-interacting electrons, namely $Z^2/2$. This became one
of the long standing open questions of mathematical physics (see,
e.g., Lieb \cite{Lieb1980} and Simon \cite[Problem 10B]{Simon1984} and
was eventually proven by Siedentop and Weikard
\cite{SiedentopWeikard1986,SiedentopWeikard1987U,SiedentopWeikard1987O,SiedentopWeikard1988,SiedentopWeikard1989}
(upper and lower bound) and Hughes \cite{Hughes1986,Hughes1990} (lower
bound) and later extended in various ways.

In Scott's spirit Lieb \cite{Lieb1981} conjectured that also
the density on the scale $Z^{-1}$ is given by the density of the Bohr
atom. This and refinements thereof were proven by Iantchenko et al
\cite{Iantchenkoetal1996,Iantchenko1997,IantchenkoSiedentop2001}. Recently
Ivrii \cite{Ivrii2019} outlined an extension.

All these results, although mathematically correct, suffer from a
serious defect viewed from a physical perspective: in the limit of
large atomic numbers $Z$ the innermost electrons are attracted more
and more to the nucleus. The ground state energy of such an electron
is even in non-relativistic quantum mechanics already $-Z^2/2$. By the
virial theorem the kinetic energy of the electron is $Z^2/2$. This
means that the corresponding classical velocity is $Z$ in atomic
units. This compares to the velocity of light $c$ which is $137$, a
dimensionless constant. Thus, say for uranium, $Z=92$, the velocity of
the innermost electrons is a substantial fraction of the velocity of
light. In other words, the limit of large $Z$ renders a
non-relativistic treatment questionable. A relativistic treatment is
required. Comparing the energies of those electrons substantiates this
view as well: the binding energy of the innermost electron of uranium
is $-4232\ Ha$ nonrelativistically compared with $-8074\ Ha$ for the
Dirac equation, i.e., almost a doubling. Schwinger \cite{Schwinger1980}
made this intuition quantitative and predicted a lowering of the
non-relativistic Scott correction.

Analogously to the non-relativistic strong Scott conjecture by Lieb,
one might predict, that the density close to the nucleus, i.e., on the
scale $Z^{-1}$, behaves in a relativistic model -- after suitable
renormalization -- like the sum of the absolute square of the
relativistic hydrogen orbitals.

To prove such statements on the ground state energy and density
starting from a microscopic model faces, however, a fundamental
problem. The physically recognized starting point should be quantum
electrodynamics. However even the most basic mathematical objects like
the underlying Hilbert space and its Hamiltonian are unknown.

But also the straightforward generalization to a multiparticle Dirac
operator -- replacing the Laplacian acting on the $n$-th particle by a
free Dirac operator -- leads to unphysical predictions. Even if the
Hamiltonian might be extended to a self-adjoint operator as recently
shown by Oelker \cite{Oelker2019} for two electrons, it leads to a
spectrum which is the whole real line and dissolution of bound states,
a fact that Brown and Ravenhall \cite{BrownRavenhall1951} observed and
is known as Brown-Ravenhall disease or continuum dissolution (Sucher
\cite{Sucher1980}). (See also Pilkuhn \cite[Section 3.7]{Pilkuhn2005}
for a review.)

Faced with this difficulty, various models were developed ranging from
straightforward quantization of the classical relativistic Hamilton
function -- which can be traced back to Chandrasekhar
\cite{Chandrasekhar1931} -- to Hamiltonians derived by physical
arguments from quantum electrodynamics like the so-called no-pair
Hamiltonians. All of those models have a critical coupling
$\gamma=\alpha Z$ at which the energy changes from being bounded to
unbounded from below (with $\alpha:=1/c$, the Sommerfeld fine
structure constant). For subcritical coupling constant the Friedrichs
extension yields a natural self-adjoint realization of the
operator. All of them show also the above mentioned lowering of the
energy.

The simplest of those models, the Chandrasekhar operator, is relatively
well studied mathematically. In fact a formula for the lowering of the
Scott term was proven by Solovej et al \cite{Solovejetal2008} and
Frank et al \cite{Franketal2008}. Moreover, recently the strong Scott
conjecture for the Chandrasekhar operator was proven as well (Frank et
al \cite{Franketal2019P}).  However, it is known that the
Chandrasekhar operator yields energies that are much too low. In fact
the really heavy elements like uranium cannot by treated at the
physical value of the fine structure constant, since $\alpha Z$
exceeds already $2/\pi$, the critical Chandrasekhar coupling constant.

The situation is improved for no-pair operators. (See Sucher
\cite{Sucher1980,Sucher1984,Sucher1987}; for a textbook discussion see
Pilkuhn \cite{Pilkuhn2005}.) Already the simplest, the Brown-Ravenhall
operator, also called no-pair operator in the free picture, raises the
energy and the critical coupling constant $2/(2/\pi+\pi/2)$ covers all
known elements at the physical value of the fine structure constant. A
corresponding formula for the Scott correction was obtained in
\cite{Franketal2009}. Nevertheless, its energies are still too low. --
A convergence result for the density on the scale $Z^{-1}$ is not known.

Chemical accuracy is obtained when the external field is included in
the definition of the state space. The corresponding operator is
called the no-pair operator in the Furry picture. A formula for the
Scott correction was proven by Handrek and Siedentop
\cite{HandrekSiedentop2015}. (The same formula should be also true
when the mean field in the sense of Mittleman \cite{Mittleman1981} is
taken into account. This, however, is so far only know in Hartree-Fock
approximation when the involved projection is given by the Dirac-Fock
operator (Fournais et al \cite{Fournaisetal2019}).) A formula for the
ground state density, however, is still missing. It is the purpose of
this paper to close this gap and to prove the strong Scott conjecture
for the no-pair operator in the Furry picture.

\section{Definitions \& main results \label{s:defresults}}
We begin with some preparatory notations which will allow to define
the no-pair Hamiltonian in the Furry picture of atoms with nuclear
charge $Z$ and $N$ electrons. We will use atomic units throughout,
i.e., the rationalized Planck constant, the elementary charge, and the
mass of the electron are all one. The energy will depend, though,
besides $Z$ and $N$, also on the velocity of light $c$. For our
purposes it is also convenient to introduce $\gamma:=\alpha Z=Z/c$.

We write $\bp:=(1/\ri)\nabla$ for the momentum operator and
\begin{align*}
  \balpha=\left(\begin{array}{cc}
            0 & \bsigma\\
            \bsigma & 0
          \end{array}\right),\hspace{1em}
  \beta=\left(\begin{array}{cc}
            \one_{\cz^2} & 0\\
            0 & -\one_{\cz^2}
          \end{array}\right),
\end{align*}
with $\bsigma=(\sigma_1,\sigma_2,\sigma_3)$ the three Pauli matrices in
standard representation for the four Dirac matrices.

We write
\begin{equation}
  \label{eq:defcoulombdirac}
 D_{c,Z}:= c\balpha\cdot \bp+c^2\beta-\frac{Z}{|x|} \quad \text{in}\ L^2(\rz^3:\cz^4)
\end{equation}
for the one-electron Dirac operator defined in the sense of Nenciu
\cite{Nenciu1976} (see also
\cite{Schmincke1972,Wust1975,KlausWust1978}), i.e., with form domain
$H^{1/2}(\rz^3:\cz^4)$ assuming $\gamma=Z/c\in(-1,1)$.

Note 
\begin{equation}
  \label{eq:cdrescaled}
  D_{c,Z}\cong c^2D_{1,Z/c}
\end{equation}
under the scaling $x\to x/c$. For the latter we drop the first index
and introduce the abbreviation
\begin{equation}
  \label{eq:dgamma}
  D_\gamma:=D_{1,\gamma}.
\end{equation}
For more general electric potentials $\varphi$ allowing also for
Nenciu's method to define the Dirac operator, we write
\begin{equation}
  \label{eq:dphi}
  D_\gamma(\varphi):=D_\gamma-\varphi.
\end{equation}

Brown and Ravenhall's basic suggestion \cite{BrownRavenhall1951} was
to implement Dirac's idea \cite{Dirac1930A} of a filled Dirac sea
which is inaccessible to physical electrons by requiring that the
state space of an electron is the positive spectral subspace of a
suitably chosen Dirac operator; in fact they suggest the positive
spectral subspace of the free Dirac operator $D_{c,0}$. Later other
choices were suggested (see Sucher \cite{Sucher1980} for more
details). A particular interesting choice is the so called Furry
picture, where the Dirac operator defining the state space is
$D_{c,Z}$ in the atomic case. It is known that the Furry picture
produces numerical values of chemical accuracy. (This choice is named
after Furry, who with Oppenheimer \cite{FurryOppenheimer1934} already
introduced the corresponding splitting of the electron and positron
space in quantum electrodynamics.)

In this paper we will focus on the Furry picture.  To be explicit, the
underlying Hilbert space is
$$
\gH_{c,Z}:=\underbrace{\left[\chi_{(0,\infty)}(D_{c,Z})\right]}_{=:\Lambda_{c,Z}}(L^2(\rz^3:\cz^4)).
$$
By Nenciu's above result,
$$
\Lambda_{c,Z}(\cS(\rz^3:\cz^4))\subseteq H^{1/2}(\rz^3:\cz^4)
$$
and dense in $\gH_{c,Z}$.

The energy $\cE_{c,Z,N}$ of an atom with nuclear charge $Z$ and $N$
electrons in the state
$\psi\in \bigwedge_{\nu=1}^N \Lambda_{c,Z}\cS(\rz^3:\cz^4)$ is
\begin{equation}
  \label{eq:manyfurry}
  \cE_{c,Z,N}[\psi]:= (\psi,\left(\sum_{\nu=1}^N((D_{c,Z}-c^2)_\nu+\sum_{1\leq \nu<\mu\leq N}\frac{1}{|x_\nu-x_\mu|}\right)\psi)
\end{equation}
The quadratic form $\cE_{c,Z,N}$ is defined as long as $D_{c,Z}$ is
defined. This is certainly true, although not necessary, if
$\gamma=Z/c\in(0,1)$, an assumption which we will make throughout the
rest of the paper.  By construction, it is bounded from below and
therefore extends according to Friedrichs to a closed quadratic form
in the Hilbert space $\bigwedge_{n=1}^N \gH_{c,Z}$ with form domain
$\gQ_{c,Z,N}:=\bigwedge_{n=1}^N(\Lambda_{c,Z}(H^{\frac12}(\rz^3:\cz^4)))$. The
resulting self-adjoint operator constructed according to Friedrichs is
the Furry operator of the -- possibly ionized -- atom of atomic number
$Z$ with $N$ electrons.

We write $F_{c,Z}$ for the operator when $N=1$ and we abbreviate
$\cE_{c,Z}:=\cE_{c,Z,Z}$ and $\gQ_{c,Z}:= \gQ_{c,Z,Z}$, i.e., we drop
the third index of the functional and its domain, when $N=Z$.

In the one-particle case, it is also here sometimes convenient to
scale out the velocity of light $c$ like for the Dirac operator and
drop a factor $c^2$ with the energies. The resulting operator depends
-- like in the case of the Coulomb-Dirac operator -- only on the
quotient $\gamma=Z/c$ and is $F_{1,\gamma}$, i.e., the index pair
$c,Z$ is replaced by $1,\gamma$. In this case, we simply drop the
index $1$ and write $F_\gamma$ in analogy to $D_\gamma$. Similarly to
\eqref{eq:dphi} we introduce
\begin{equation}
  \label{eq:fphi}
  F_\gamma(\varphi)\ \text{and}\ F_{c,Z}(\varphi)
\end{equation}
as the self-adjoint operators associated with
$(f,(D_\gamma-\varphi-1)f)$ on $\Lambda_{\gamma}\cS(\rz^3:\cz^4)$ and
$(f,(D_{c,Z}-\varphi-1)f)$ on $\Lambda_{c,Z}\cS(\rz^3:\cz^4)$
whenever closable and bounded from below.

Matte and Stockmeyer \cite[Theorem 2.2]{MatteStockmeyer2010} showed that
\begin{equation}
  \label{gs}
  E_{c,Z}:=\inf\{\cE_{c,Z}[\psi]| \psi\in\gQ_{c,Z},\ \|\psi\|=1\}
\end{equation}
is assumed, i.e., a ground state -- not necessarily uniquely
determined -- exists.  Although, we neither need the state to be pure
nor exactly a minimizer, we will refrain from such generalizations,
and pick the state which occurs according to L\"uders
\cite{Luders1951} when measuring the ground state energy, namely
$$\Lambda:=(|\psi_1\rangle\langle\psi_1|+...+|\psi_D\rangle\langle\psi_D|)/D$$
where $\psi_1,...,\psi_D$ is an orthonormal basis of the ground state
space of $\cE_{c,Z}$ (with $N=Z$). We will denote the corresponding
spin-summed one-particle density by
\begin{equation}
  \label{gsd}
  \rho(x) := \frac ND \sum_{d=1}^D\sum_{\sigma=1}^4\int_{\Gamma^{(N-1)}} |\psi_d(x,\sigma;y)|^2\,\dy
\end{equation}
where $y\in\Gamma^{N-1}$ with $\Gamma:=\rz^3\times\{1,2,3,4\}$ are
space-spin variables. Moreover, $\rd y$ is the corresponding measure,
i.e., integration in the space variable and summation in the spin
variable.

A refinement is to consider the density in angular momentum channels,
more accurately in channels of spin-orbit coupling, labeled by
\begin{equation}
  \label{zp}\kappa\in\zp:=\gz\setminus\{0\}
\end{equation}
(see Appendix \ref{Anhang1} for more details). It is also convenient
to introduce
\begin{equation}
  \label{l}
  j_\kappa:=|\kappa|-\tfrac12\ \text{and}\
  \ell_\kappa:= j_\kappa-\tfrac12\sgn(\kappa)= |\kappa|-\theta(\kappa),
\end{equation}
the quantum numbers of total and orbital angular momentum all
determined by $\kappa$.

The density at a point $x\in\rz^3$ in channel
$\kappa\in\zp$ of the ground state $\Lambda$ of $\cE_{c,Z}$ is
\begin{equation}
  \label{eq:furryangdensity}
  \rho_{\kappa}(x):= {N\over4\pi D}\sum_{d=1}^D\sum_{\sigma\in\{+,-\}}\sum_{m=-j_\kappa}^{j_\kappa}\int_{\Gamma^{N-1}}\rd y
  \left|\sum_{\tau=1}^4\int_{\mathbb{S}^2}\rd \omega\,
    \overline{\Phi_{\kappa,m}^\sigma(\omega,\tau)}
    \psi_d(|x|\omega,\tau,y)\right|^2
\end{equation}
where $\Phi^\sigma_{\kappa,m}$ are spherical Dirac spinors
\eqref{fi}. (See \eqref{a1} for the relation to $\Pi_\kappa$.)

Note that $\rho=\sum_{\kappa\in\zp}\rho_\kappa$ for the state $\Lambda$. For
general states, the left side needs an additional spherical
average. The functions $\rho_\kappa$ and $\rho$ are the objects of
interest of this work. When appropriately rescaled, we will study
their convergence as $Z\to\infty$ and $Z/c$ is fixed. The objects
which will turn out to be the limits are introduced now.

We write $\psi_{n,\kappa,m}$ for the orthonormal eigenfunctions of
$D_\gamma$, i.e., 
\begin{equation}
  \label{eigenfunktion}
  D_\gamma\psi_{n,\kappa,m}=\lambda_{n,\kappa}\psi_{n,\kappa,m}
\end{equation}
suppressing the dependence of $\gamma$ in both the eigenvalues and
eigenfunctions. The corresponding eigenvalue problem was solved by
Gordon \cite{Gordon1928} and Darwin \cite{Darwin1928}. (See also Bethe
\cite[Formula (9.29)]{Bethe1933} or \cite[Formula
(7.140)]{Thaller1992} for textbook treatments.) 

The density of a Bohr atom for a given $\gamma$ in channel $\kappa$ is then
defined by
\begin{equation}
  \label{eq:2.12}
  \rho_{\kappa}^H(x) := \sum_{n=\theta(-\kappa)}^\infty\sum_{m=-j_\kappa}^{j_\kappa}
  \sum_{\sigma=1}^4|\psi_{n,\kappa,m}(x,\sigma)|^2;
\end{equation}
the total hydrogenic density is 
\begin{equation}
  \label{rhoh}
  \rho^H(x) := \sum_{\kappa\in\zp}\rho_\kappa^H(x).
 \end{equation}
 Of course, this is only well defined, if the right sides of
 \eqref{eq:2.12} and \eqref{rhoh} converge which we will show outside
 the origin in Theorem \ref{existencerhoh}. Moreover, we will study
 its behavior as $x\to0$ and $x\to\infty$.

 Finally, since we show convergence in a weak -- although in fact in
 the radial variable rather strong -- sense, we need to specify the
 test functions. The test functions can be written as $U=U_1+U_2$
 where $U_1$ may have a Coulomb singularity at $r=0$ and $U_2$ decays
 sufficiently fast as $r\to\infty$. More precisely, $U_2$ is going to
 belong to the test function spaces $\cD_\gamma^{(0)}$ and $\cD$ used
 by Frank et al \cite{Franketal2019P}. For the convenience of the
 reader we give their definition also in Appendix
 \ref{s:testfunctions}, in particular \eqref{eq:deftestnod0} and
 \eqref{eq:deftest0}.  As an example, we mention that if the test
 function obeys
\begin{equation}
  \label{eq:example}
  |U(r)| \leq \const \left(r^{-1}\one_{\{r\leq1\}} + r^{-\alpha}\one_{\{r>1\}}\right),
\end{equation}
then $U=U_1+U_2$ with $U_1\in r^{-1}L^\infty_c ([0,\infty))$
bounded and compactly supported and $U_2$ belongs to
$\cD_\gamma^{(0)}$, if $\alpha>1$. It is in
$\cD_\gamma^{(0)}\cap\cD_\gamma$, if $\alpha>3/2$. (The index $c$
denotes, as usual, compact support.)

\begin{theorem}[Convergence of the angular momentum decomposed density]
  \label{convfixedl}
  Fix $\kappa\in\zp$, and $U=U_1+U_2$ with
  $U_1\in r^{-1}L^\infty_c ([0,\infty))$ and
  $U_2\in\cD_\gamma^{(0)}$. Then, with $\gamma=Z/c\in(0,1)$ fixed,
  $$
  \lim_{Z\to\infty}\int_{\rz^3} c^{-3}\rho_{\kappa}(c^{-1}x)U(|x|)\,\rd x
  = \int_{\rz^3}\rho_{\kappa}^H(x)U(|x|)\,\rd x.
  $$
\end{theorem}

\begin{theorem}[Convergence of the density]
  \label{convalll}
  Let $U=U_1+U_2$ with
  $U_1\in r^{-1}L^\infty_c ([0,\infty))$,
  $U_2\in\cD\cap\cD_\gamma^{(0)}$, and $\gamma\in(0,1)$. Then
  $$
    \lim_{Z\to\infty}\int_{\rz^3} c^{-3}\rho(c^{-1}x)U(|x|)\,\dx = \int_{\rz^3} \rho^H(x)U(|x|)\,\dx.
  $$
\end{theorem}

The next result ensures that the above convergence results are not
meaningless. More precisely, we will now show that the hydrogenic densities
are finite for all $r\in\rz_+$.
To this end define for $\gamma\in[0,1]$
\begin{align}
  \label{eq:defsigmagamma}
  \begin{split}
    \sigma_\gamma = 1-\sqrt{1-\gamma^2}\in[0,1].
  \end{split}
\end{align}
Note that $\sigma_0=0$, $\sigma_1=1$, $\sigma_{\sqrt3/2}=1/2$,
$\sigma_{\sqrt{15}/4}=3/4$, and $\sigma_\gamma$ is strictly monotone
increasing.  We will denote positive constants from now on by $A$ or
$a$.  Any dependence on some parameter is going to be denoted by a
corresponding subscript. Moreover, positive constants may vary from
line to line but are still going to be denoted by the same letter.

\begin{theorem}[Existence of $\rho_\kappa^H$ and $\rho^H$]
  \label{existencerhoh}
  Let $1/2<s\leq3/4$, if $\gamma\in(0,\sqrt{15}/4)$ and
  $1/2<s<3/2-\sigma_\gamma$, if $\gamma\in[\sqrt{15}/4,1)$.
  Then there is a constant $A_{s,\gamma}>0$ such that
  for all $\kappa\in\zp$ and $x\in\rz^3\setminus\{0\}$
  \begin{multline*}
    \rho_{\kappa}^H(x)
    \leq A_{s,\gamma} {|\kappa|^{1-4s}\over|x|^2}\\
    \times\left[\left(\frac{|x|}{|\kappa|}\right)^{2s-1}\one_{\{|x|\leq |\kappa\}}+\left(\frac{|x|}{|\kappa|}\right)^{4s-1}\one_{\{|\kappa|\leq |x|\leq |\kappa|^2\}}
      + |\kappa|^{4s-1}\one_{\{|x|\geq |\kappa|^2\}}\right].
  \end{multline*}
  Moreover, for any $\varepsilon>0$ and $x\in\rz^3\setminus\{0\}$,
  there are constants $A_{\gamma,\varepsilon},A_\gamma>0$ such that
  \begin{align}
    \label{eq:boundrhoh}
    \rho^H(x)\leq
    \begin{cases}
      A_\gamma |x|^{-3/2} & \text{if}\ \gamma\in(0,\sqrt{15}/4]\\
      A_{\gamma,\varepsilon}\left(|x|^{-2\sigma_\gamma-\varepsilon}\one_{\{|x|\leq1\}}
        + |x|^{-3/2}\one_{\{|x|>1\}}\right) & \text{if}\ \gamma\in(\sqrt{15}/4,1)
    \end{cases}.
  \end{align}
\end{theorem}

Some remarks on the above results are in order.

(1) The corresponding convergence results and pointwise bounds on the
hydrogenic densities were recently proven for Chandrasekhar atoms by
Frank et al \cite{Franketal2019P}.
The classes of admissible test functions are the same in both models,
i.e., the test functions may have Coulomb singularities at the origin,
but delta functions are not allowed, i.e., we were not able to
prove pointwise convergence of the densities.

For a comparison between the results of Iantchenko et al
\cite{Iantchenkoetal1996} in the non-relativistic case with those that
were obtained for the above two relativistic models, we refer to the
discussion after \cite[Theorem 2]{Franketal2019P}.

(2) As in \cite{Franketal2019P} we show that the hydrogenic density is
finite for all $x\in\rz^3$ and obtain a pointwise upper bound with a
similar asymptotic behavior for small and large distances to the
nucleus.  Although we are lacking a corresponding lower bound and the
constant appearing in Theorem \ref{existencerhoh} is implicit and
presumably far from sharp, we believe that the dependence on $r$ is
optimal: on the one hand, relativistic effects should play a minor
role for $r\gg1$ which is reflected in the $r^{-3/2}$-decay of
$\rho^H$. In fact, Heilmann and Lieb \cite{HeilmannLieb1995} proved in
the non-relativistic case that the density decays like
$(2^{3/2}/(3\pi^2))\gamma^{3/2}|x|^{-3/2}+o(|x|^{-3/2})$ as
$x\to\infty$.  Recalling that the Thomas-Fermi density satisfies
$\rho_Z^{\mathrm{TF}}(x)=Z^2\rho_1^{\mathrm{TF}}(Z^{1/3}x)\sim(Z/|x|)^{3/2}$
as $x\to0$, the bounds on $\rho^H(x)$ for large $|x|$ indicate that
there is a smooth transition between the quantum length scale $Z^{-1}$
and the Thomas-Fermi length scale $Z^{-1/3}$.  Note also that a lower
bound of the form $\rho^H(x)\geq A_\gamma |x|^{-3/2}\one_{\{|x|\geq1\}}$
would suggest that the function space $\cD_\gamma^{(0)}\cap\cD$ is
optimal in the sense that it covers functions that decay like
$|x|^{-3/2-\varepsilon}$, see \eqref{eq:example}.

On the other hand, the behavior for small $r$ seems best possible for
$\gamma\geq\sqrt{15}/4$, except for the lack of a corresponding lower bound
and the arbitrary small $\varepsilon$ appearing in \eqref{eq:boundrhoh}.
The main reason for this belief is the behavior of the radial part of the
hydrogenic ground state wave function at the origin
$$
|\psi_{n=0,\kappa=\pm1,m}(x)| \sim |x|^{\sqrt{1-\gamma^2}-1}
= |x|^{-\sigma_\gamma}, \  \ m=-j_\kappa,...,j_\kappa.
$$
The formula reveals in particular, that the singularity of the
hydrogenic density is only generated by the eigenfunctions with
$|\kappa|=1$, since $\psi_{0,\kappa,m}$ has no singularity at
the origin for any $|\kappa|\geq2$. This observation supports our
claim for the small $r$ behavior of $\rho^H$ for
$\gamma\geq\sqrt{15}/4$.  However, the formula also shows that our
bound for $\gamma<\sqrt{15}/4$ cannot be optimal, since it does not
depend on $\gamma$ at all.  As in the Chandrasekhar case, this
limitation is of technical nature and comes from the restriction
$\sigma_\gamma\leq3/4=\sigma_{\sqrt{15}/4}$.  Ultimately, the behavior
of the eigenfunctions with $|\kappa|=1$ and $\gamma=1$ suggests that
the admissible singularities of our test functions are optimal. This
is also expected in view of Kato's inequality since singularities
which are more severe than Coulomb cannot be controlled by kinetic
energy anymore.

Although the eigenfunctions $\psi_{n,\kappa,m}$ are explicitly known,
the explicit summation of their absolute squares analogously to Heilmann and
Lieb \cite{HeilmannLieb1995} in the non-relativistic case is an
open question. An answer would most likely allow for a more detailed
study of the properties of $\rho^H$.

(4) The basic idea behind the proof of the convergence result is a
linear response argument which was already used by Baumgartner
\cite{Baumgartner1976}, Lieb and Simon \cite{LiebSimon1977},
Iantchenko et al \cite{Iantchenkoetal1996} and Frank et al
\cite{Franketal2019P}.  We first estimate the difference of the
expectation values of the appropriately perturbed and unperturbed
many-body Hamiltonians in the unperturbed ground state by the spectral
shift between the correspondingly perturbed and unperturbed hydrogenic
one-particle operators.  Then, we use the generalized Feynman-Hellmann
theorems \cite[Theorem 13, Proposition 14]{Franketal2019P} to
differentiate the sum of the negative eigenvalues of the perturbed
hydrogenic operator.  The main difficulty consists in verifying the
assumptions of these theorems. In particular, we will show that the
test function $U$ satisfies a certain ``relative trace class
condition'' with respect to $F_\gamma$ in channel $\kappa$.  To be
definite introduce the notation
\begin{equation}
  \label{ps}
  \tr_\kappa A:= \tr(\Pi_\kappa A)
\end{equation}
when $\Pi_\kappa A$ is trace class.

For convenience we also introduce the abbreviation for the Furry
operators in angular momentum channel $\kappa$
\begin{equation}
  \label{radialeops} f_{\gamma,\kappa}:=F_\gamma|_\ghpk,\ f_{\gamma,\kappa}(\varphi):= F_\gamma(\varphi)|_\ghpk
\end{equation}
which we will freely use here and later.

Then our claim is that for $s>1/2$ and $\kappa\in\zp$
$$
\tr (f_{\gamma,\kappa}+1)^{-s}\Pi_\kappa^+U\Pi^+_\kappa(f_{\gamma,\kappa}+1)^{-s}<\infty.
$$
 As in the Chandrasekhar case, $s>1/2$ is crucial, since
$(1+k)^{-1}\notin L^1(\rz_+,\rd k)$.

The general strategy to prove the above and similar assertions, is to
roll them back to those involving Chandrasekhar operators where they
are known to hold \cite[Corollary 20]{Franketal2019P}.  A main new
technical contribution is to show that the Chandrasekhar and the Furry
operators are comparable: Corollary \ref{hardydomcor}
shows that one can compare $\Lambda_\gamma(|p|+1)^s\Lambda_\gamma$
with $(F_\gamma+1)^s$ which will be an important tool.

\section{Applying the Feynman-Hellmann theorem in the Furry picture:
  the case of fixed $\kappa$}

We will use the abstract version of the Hellmann-Feynman theorem by
Frank et al \cite[Theorem 13 and Proposition 14]{Franketal2019P}. To
be self-contained we recall these results here. The first one will be
used to handle the Coulomb singularity, whereas the second one will
handle the local singularities of the test potential.

We write $A_-=-A\chi_{(-\infty,0)}$ and denote by
$\gS^1$ the set of trace class operators and by $\gS^2$ the set of
Hilbert-Schmidt operators.
\begin{proposition}
  \label{diffgen0}
  Assume that $A$ is a self-adjoint operator on a Hilbert space $\gH$
  with $A_-\in\gS^1(\gH)$ and $B$ is non-negative and relatively form bounded
  with respect to $A$. Furthermore, assume that there is
  $1/2\leq s\leq 1$ and $M>-\inf\sigma(A)$ such that
  \begin{equation}
    \label{eq:traceclassdelta}
    (A+M)^{-s} B (A+M)^{-s}\in\gS^1(\gH)
  \end{equation}
  and
  \begin{equation}
    \label{eq:relbounddelta}
    \limsup_{\lambda\to 0} \left\| (A+M)^{s} (A-\lambda B+M)^{-s} \right\| <\infty .
  \end{equation}
  Then the one-sided derivatives $D^\pm S$ of
  $$
  \lambda\mapsto S(\lambda):=\tr(A-\lambda B)_-
  $$
  satisfy
  $$
  \tr B\chi_{(-\infty,0)}(A)=D^-S(0)\leq D^+S(0)=\tr B\chi_{(-\infty,0]}(A).
  $$
  In particular, $S$ is differentiable at $\lambda=0$, if and only
  if $B|_{\ker A}=0$.
\end{proposition}

\begin{proposition}
  \label{diffgen}
  Assume that $A$ is self-adjoint with $A_{-} \in\mathfrak{S}^1$ and
  $B$ is non-negative, and let $1/2<s\leq1$. Assume that there is an
  $s'<s$ such that for some $M>-\inf\sigma(A)$,
  \eqref{eq:traceclassdelta} holds, and that for some $a>0$
  \begin{equation}
    \label{eq:increasedrelbdd}
    B^{2s} \leq a(A+M)^{2s'}.
  \end{equation}
  Then $B$ is form bounded with respect to $A$ with form bound $0$ and
  the conclusions in Proposition \ref{diffgen0} hold.
\end{proposition}

We recall the following two observations.

1. By \eqref{eq:traceclassdelta},
$$
\tr(B\chi_{(-\infty,0]}(A))=-\tr\left[((A+M)^{-s}B(A+M)^{-s})(\chi_{(-\infty,0]}(A)(A+M)^{2s})\right]<\infty,
$$
since $\chi_{(-\infty,0]}(A)(A+M)^{2s}$ is bounded.
Hence, also $D^+S(0)<\infty$.

2. If the bottom of the essential spectrum of $A$ is strictly
positive, the result recovers the classical Feynman-Hellmann theorem.
The point is that the formulae remain valid even, if
$\inf\sigma_{\mathrm{ess}}(A)=0$, i.e., the case where perturbation
theory is not directly applicable.

In the application of the two propositions above, the underlying
Hilbert space is $\ghpk$, $A$ will be the Furry
operator $F_{\gamma}$ restricted to this space, and
$B = \Lambda_\gamma(U\otimes\one_{\cz^4})\Lambda_\gamma$ also
restricted to this space plays the role of the test potential.

We recall some basic facts about $F_\gamma$.
\begin{lemma}
  \label{hydrogen}
  Let $\kappa\in\zp$ and
  $\gamma\in[0,1)$.  Then $F_\gamma\geq\sqrt{1-\gamma^2}-1$,
  $0\notin\sigma_\mathrm{pp}(F_\gamma)$, and $F_\gamma+1$ has a
  bounded inverse. Moreover,
  $\tr_\kappa (F_\gamma)_-<\infty$.
\end{lemma}
\begin{proof}
  The fact that
  $\sigma_\mathrm{ess}(D_\gamma)=\rz\setminus(-1,1)\cap
  \sigma_\mathrm{pp}(D_\gamma)=\emptyset$ is a standard
  consequence of the virial theorem proven by Kalf \cite{Kalf1976} for
  all $\gamma\in(-1,1)$. In particular the eigenvalues of $D_\gamma$
  are all given by Sommerfeld's eigenvalue formula
  \begin{align}
    \label{eq:eigenvalue}
    \lambda_{n,\kappa} = \left(1+{\gamma^2\over \left(n+\sqrt{\kappa^2-\gamma^2}\right)^2}\right)^{-1/2}
  \end{align}
  with $(\kappa,n)\in(-\nz\times\nz)\cup(\nz\times\nz_0)$ (Sommerfeld
  \cite{Sommerfeld1916}, Gordon \cite{Gordon1928}, and Darwin
  \cite{Darwin1928}). In particular the lowest eigenvalue is
  $\sqrt{1-\gamma^2}$ and $\sum_{n}(\lambda_{n,\kappa}-1)$ is
  absolutely summable for each fixed $\kappa$.
\end{proof}
For a textbook discussion of $D_\gamma$, we refer to
Bethe \cite{Bethe1933} and Thaller \cite{Thaller1992}, in particular
\cite[p. 314f]{Bethe1933} and \cite[Sections 7.4.2 and
7.4.5]{Thaller1992} for the discussion of the point spectrum.

The following two propositions show the applicability of
Propositions \ref{diffgen0} and \ref{diffgen}. They are the keys to
prove Theorem \ref{convfixedl}.

\begin{proposition}
  \label{diffu2}
  Let $\gamma\in(0,1)$, $\kappa\in\zp$, and
  $0\leq U\in\cD_\gamma^{(0)}$.  Then
  $$
  \lambda\mapsto \tr f_{\gamma,\kappa}(\lambda U)_-
  $$
  is differentiable at $\lambda=0$ with derivative
  $\int_{\rz^3}\rho_\kappa^H(x) U(|x|)\,\dx$.
\end{proposition}

\begin{proposition}\label{diffu1}
  Let $\gamma\in(0,1)$, $\kappa\in\zp$, and
  $0\leq U \in r^{-1} L^\infty_c ([0,\infty))$.  Then
  $$
  \lambda\mapsto
  \tr f_{\gamma,\kappa}(\lambda  U)_-
  $$
  is differentiable at $\lambda=0$ with derivative
  $\int_{\rz^3} \rho_\kappa^H(x) U(|x|)\,\dx$.
\end{proposition}

Note that Propositions \ref{diffu2} and \ref{diffu1}  imply
$\int_{\rz^3}\rho_\kappa^H(x)U(|x|)\dx<\infty$.  In particular, these
results show that for $\gamma<1$ and $R>0$,
$$
\int_{|x|<R} \rho_\kappa^H(x) |x|^{-1}\,\dx < \infty.
$$
In fact, there is also a simple, direct proof of this, even when
$R=\infty$: based on a computation of Burke and Grant
\cite{BurkeGrant1967}, Handrek and Siedentop \cite[Lemma
2]{HandrekSiedentop2015} show that the potential energy of hydrogenic
eigenfunctions satisfies
$$
(\psi_{n,\kappa,m},\frac{\gamma}{|x|}\psi_{n,\kappa,m})
\leq\frac{a_{\gamma_0}\gamma^2}{(n+|\kappa|)^2}.
$$
Clearly, the right side is summable in $n$ and -- trivially -- in $m$.

Propositions \ref{diffu2} and \ref{diffu1} will be deduced from
Propositions \ref{diffgen} and \ref{diffgen0} respectively.  To verify
their assumptions, we will first reduce the problem to the scalar
Chandrasekhar operator and then use \cite{Franketal2019P}.

In this and the next section, we will often use the
Davis-Sherman inequality (Davis \cite{Davis1957}, see also
Carlen \cite[Theorem 4.19]{Carlen2010}).
It says that for all operator convex functions $f$ and all orthogonal
projections $P$, the form inequality
$$
Pf(PAP)P \leq Pf(A)P
$$
holds for all self-adjoint operators $A$.  If, moreover, $f(0)=0$, then
\begin{equation}
  \label{eq:SD}
  f(PAP) \leq Pf(A)P.
\end{equation}
Indeed, for $P^\perp=1-P$, one has
$$
f(PAP) = Pf(PAP)P + P^\perp f(PAP)P^\perp,
$$
since $P$ commutes with $PAP$ and therefore with any function
$f(PAP)$. However, by the spectral theorem,
$P^\perp f(T)P^\perp = P^\perp f(P^\perp TP^\perp)P^\perp$ for any
self-adjoint operator $T$ commuting with $P^\perp$ or $P$.  This
yields (with $T=PAP$)
$$
P^\perp f(PAP)P^\perp = P^\perp f(P^\perp PAPP^\perp)P^\perp
= P^\perp f(0)P^\perp
$$
which vanishes, if $f(0)=0$. -- We will apply this, when $A$ is a
non-negative operator and $f(x)=x^{2s}$ with $s\in[1/2,1]$.


\subsection{Comparison between the Chandrasekhar and the Furry operator}

We write $D^0_\gamma:=\balpha\cdot\bp-\gamma|x|^{-1}$ for the massless
Coulomb-Dirac operator (which is defined as in Section
\ref{s:defresults}, Nenciu \cite{Nenciu1976}).  The following lemma
gives a comparison between $|p|^s$ and $|D_\gamma^0|^s$ as operators
in $L^2(\rz^3:\cz^4)$.
\begin{lemma}[Frank et al {\cite[Corollary 1.8]{Franketal2019}}]
  \label{hardydom}
  Let $\gamma\in[0,1)$ and $s\in(0,1]$.
  Then there exists an $A_{s,\gamma}<\infty$ such that
  $$
  |D^0_\gamma|^{2s} \leq A_{s,\gamma} |p|^{2s}
  $$
  If, additionally, $s<3/2-\sigma_{\gamma}$, then there is
  an $a_{s,\gamma}>0$ such that
  $$
  |D^0_\gamma|^{2s} \geq a_{s,\gamma} |p|^{2s}
  $$
\end{lemma}

From Lemma \ref{hardydom}, we deduce
\begin{corollary}
  \label{hardydomcor}
  Let $\gamma\in[0,1)$ and $M\geq0$.
  If $0<s<\min\{3/2-\sigma_\gamma,1\}$, then
  \begin{align}
    \label{eq:hardydomcorform}
    \Lambda_\gamma  (|p|^{2s}+M)\Lambda_\gamma
    \leq (1-\gamma^2)^{-s}(a_{s,\gamma}^{-1}+M) (F_\gamma+1)^{2s}.
  \end{align}
  Moreover, if $0<s\leq1$, then
  $$
  (F_\gamma+1)^{2s}
  \leq 2^s(1+4\gamma^2)^s \Lambda_\gamma\left( |p|+(1+4\gamma^2)^{-\frac12} \right)^{2s}\Lambda_\gamma.
  $$
\end{corollary}

\begin{proof}
  We begin with the first claim. Since $\sqrt{1-\gamma^2}$ is the
  lowest positive spectral point of $D_\gamma$, it suffices to show
  the claim for $M=0$.  Next, note that
  $$
  (D_\gamma)^2 \geq (1-\gamma^2) |D_\gamma^0|^2
  $$
  by Morozov and M\"uller \cite[Proof of Corollary
  I.2]{MorozovMuller2017}. By operator monotonicity of $x\mapsto x^s$
  with $s\in(0,1]$, and Lemma \ref{hardydom} we have
  \begin{align*}
    \Lambda_\gamma|p|^{2s}\Lambda_\gamma
    & \leq \Lambda_\gamma a_{s,\gamma}^{-1}|D_\gamma^0|^{2s}\Lambda_\gamma
    \leq (1-\gamma^2)^{-s}a_{s,\gamma}^{-1}\Lambda_\gamma|D_\gamma|^{2s}\Lambda_\gamma\\
    & = (1-\gamma^2)^{-s}a_{s,\gamma}^{-1}\left(\Lambda_\gamma D_\gamma\Lambda_\gamma\right)^{2s}
  \end{align*}
  where the last step is obvious by the spectral theorem.

  We turn to the second inequality. First we note that the left side is equal to $\Lambda_\gamma |D_\gamma|^{2s}\Lambda_\gamma$, i.e., it suffices to prove the stronger inequality
  $$|D_\gamma|^{2s}\leq 2^s(1+4\gamma^2)^s(|p|+(1+4\gamma^2)^{-1/2})^{2s}$$
    By operator monotonicity of roots, it is enough to prove the claim
    for largest occurring $s$, namely $s=1$. This, however, follows by
    first using the Schwarz inequality and then Hardy's inequality
  $$
    |D_\gamma|^2  \leq 2\left(p^2+1+\frac{\gamma^2}{|x|^2}\right)
    \leq 2 (|p|^2(1+4\gamma^2)+1) \leq 2(1+4\gamma^2)\left(|p|+{1\over\sqrt{1+4\gamma^2}}\right)^2
    $$
    where the last step is obvious.
\end{proof}

We introduce the following restricted operators in
$\gh_{\kappa}$,
\begin{align*}
  p_{\ell_\kappa} & := \sqrt{-\Delta}|_{\gh_{\kappa}},\\
  C_{\ell_\kappa} & := \left(\sqrt{-\Delta+1}-1\right)|_{\gh_{\kappa}}.
\end{align*}
The corresponding radial operators in $L^2(\rz_+:\cz,\dr)$ are
\begin{align}
  \label{eq:radialkinetic}
  p_{\ell_\kappa}^{(r)} & := \sqrt{ - \frac{\rd^2}{\dr^2} + \frac{\ell_\kappa(\ell_\kappa+1)}{r^2}},\\
  C_{\ell_\kappa}^{(r)} & :=\sqrt{- \frac{\rd^2}{\dr^2} + \frac{\ell_\kappa(\ell_\kappa+1)}{r^2} + 1}-1.
\end{align}
We note that the bounds of Corollary \ref{hardydomcor} continue to
hold in each $\gh_{\kappa,m}$.
Recall that any element $f\in\gh_{\kappa,m}$ is of the form
$$
f(x)=\sum_{\sigma\in\{+,-\}} |x|^{-1}
f^\pm(|x|)\Phi_{\kappa,m}^\sigma(x/|x|)
$$
where the $\Phi_{\kappa,m}^\pm$ are defined in \eqref{fi} and
$f^\pm\in L^2(\rz_+)$.  Both $D_\gamma$ and $|p|$ leave the spaces
$\gh_{\kappa,m}$ invariant, i.e., they commute with the projection
$\Pi_{\kappa,m}$. Indeed, for
$f\in\gh_{\kappa,m}\cap H^{1}(\rz^3:\cz^4)$ and
$g\in\gh_{\kappa',m'}\cap H^{1}(\rz^3:\cz^4)$, one has
\begin{align*}
  & ( f,|p| g)_{L^2(\rz^3:\cz^4)}
    = (( f^+,p_{\ell_\kappa}^{(r)} g^+)_{L^2(\rz_+:\cz)} + ( f^-,p_{2j_\kappa-\ell_\kappa}^{(r)} g^-)_{L^2(\rz_+:\cz)})\delta_{\kappa,\kappa'}\delta_{m,m'}
\end{align*}
and
\begin{equation}
  \label{eq:radialdirac}
  \begin{split}
    & ( f,D_\gamma g)_{L^2(\rz^3:\cz^4)}\\
    & \quad = ( \left(\begin{array}{c}
                      f^+\\ f^-
                    \end{array}
    \right),\left(\begin{array}{cc}
                   1-\frac{\gamma}{r} & -\frac{\rd}{\dr}-\frac{\kappa}{r}\\
                   \frac{\rd}{\dr}-\frac{\kappa}{r} & -1-\frac{\gamma}{r}
                 \end{array}
    \right)\left(\begin{array}{c}
                   g^+\\ g^-
                 \end{array}
    \right))_{L^2(\rz_+:\cz^2)}\delta_{\kappa\kappa'}\delta_{mm'},
  \end{split}
\end{equation}
see also \cite[Formula (7.105)]{Thaller1992}.
Together with the spectral theorem, this shows that the projection
of \eqref{eq:hardydomcorform} onto $\gh_{\kappa,m}$, namely
\begin{equation}
  \label{eq:hardydomcorang0}
  \Pi_{\kappa,m}\Lambda_\gamma(|p|+M)^{2s}\Lambda_\gamma\Pi_{\kappa,m}
  \leq a_{s,\gamma} \Pi_{\kappa,m}(F_\gamma+1)^{2s}\Pi_{\kappa,m},
\end{equation}
is equivalent to
\begin{equation}
  \label{eq:hardydomcorang}
  \Lambda_\gamma(\Pi_{\kappa,m}(|p|+M)\Pi_{\kappa,m})^{2s}\Lambda_\gamma
  \leq a_{s,\gamma}(\Pi_{\kappa,m}(F_{\gamma}+1)\Pi_{\kappa,m})^{2s}.
\end{equation}
Mutatis mutandis, the equivalence holds also for the
projection onto $\gh_\kappa$.

\subsection{Trace inequalities in $\gh_\kappa$}

We recall some trace and form inequalities for functions belonging to
the spaces $\cK_{s}^{(0)}$ introduced by Frank et al
\cite{Franketal2019P}. For the convenience of the reader, we give
their definition Appendix \ref{s:testfunctions}. Frank et al
\cite{Franketal2019P} wrote the associated trace inequalities in terms
of powers of $C_\ell^{(r)}+M$. Using Plancherel's theorem, one can
rewrite them as inequalities in powers of $p_\ell^{(r)}+M$
instead. Here, we will actually formulate the inequalities in terms of
$p_\ell$.

\begin{lemma}
  \label{genReltrclassnod}
  Let $M>0$, $s\in(1/2,1]$, $\kappa\in\zp$, and
  $0\leq W \in\cK_{s}^{(0)}$. Then
  \begin{align}
    \label{eq:genReltrclassnod}
    \tr[(p_{\ell_\kappa}+M)^{-s}W(p_{\ell_\kappa}+M)^{-s}]
    \leq A_{s,\kappa,M} \|W\|_{\cK_{s}^{(0)}}.
  \end{align}
  In particular, we have in $L^2(\rz^3:\cz^4)$
  \begin{align}
    \label{eq:genSobolevDualnod}
    \Pi_{\kappa} W\Pi_{\kappa}
    \leq A_{s,\kappa,M} \|W\|_{\cK_{s}^{(0)}} (\Pi_{\kappa}(|p|+M)\Pi_{\kappa})^{2s}.
  \end{align}
\end{lemma}

\begin{proof}
  The estimate \eqref{eq:genReltrclassnod} follows from
  \begin{align*}
    \|W^{\frac12}(p_{|\kappa|}^{(r)}+M)^{-s}\|_{\gS^2(L^2(\rz_+))}^2 + \|W^{\frac12}(p_{|\kappa|-1}^{(r)}+M)^{-s}\|_{\gS^2(L^2(\rz_+))}^2
    \leq A_{s,\kappa,M} \|W\|_{\cK_{s}^{(0)}}
  \end{align*}
  (Frank et al \cite[Proposition 19]{Franketal2019P}).
  Estimate \eqref{eq:genSobolevDualnod} follows immediately
  from \eqref{eq:genReltrclassnod}.
\end{proof}

Combining Corollary \ref{hardydomcor} in each channel $\kappa$, i.e.,
\eqref{eq:hardydomcorang} and Lemma \ref{genReltrclassnod} yields a
generalization of the previous inequalities but now with respect to
the Furry operator. Using the notation $\Pi^+_\kappa$ defined in
\eqref{a2} we have
\begin{lemma}
  \label{genReltrclassHydro}
  For $\gamma\in(0,1)$, $1/2<s<\min\{3/2-\sigma_\gamma,1\}$,
  $\kappa\in\zp$, and $0\leq W \in\cK_{s}^{(0)}$
  \begin{equation}
    \label{eq:genReltrclassHydro}
    \begin{split}
      \tr[(f_{\gamma,\kappa}+1)^{-s}\Pi^+_\kappa W\Pi^+_{\kappa} (f_{\gamma,\kappa}+1)^{-s}]
      \leq A_{\gamma,s,\kappa} \|W\|_{\cK_{s}^{(0)}}.
    \end{split}
  \end{equation}
  In particular
  \begin{equation}
    \label{eq:genSobolevDualHydro}
 \Pi_{\kappa}^+W\Pi_{\kappa}^+
      \leq A_{\gamma,s,\kappa} \|W\|_{\cK_{s}^{(0)}} (f_{\gamma,\kappa}+1)^{2s}.
  \end{equation}
\end{lemma}
\begin{proof}[Proof of Proposition \ref{diffu2}]
  We apply Proposition \ref{diffgen} with $M=1$,
  $A = f_{\gamma,\kappa}$, and $B = \Pi_\kappa^+ U\Pi^+_\kappa$. Here $U\in\cK_{s}^{(0)}$
  and $U^{2s}\in\cK_{s'}^{(0)}$ with $1/2<s'<s\leq1$, if
  $\gamma<\sqrt3/2$, and $1/2<s'<s<3/2-\sigma_\gamma$, if
  $\gamma\in[\sqrt3/2,1)$.

We now verify the assumptions of Proposition \ref{diffgen}:
the assumptions on $F_{\gamma}$ in $\gH_\kappa^+$
follow from Lemma \ref{hydrogen}. In particular, since zero is not
an eigenvalue of $F_{\gamma}$, the right and left derivative
agree at $\lambda=0$.

Since $U\in\cK_{s}^{(0)}$, we have by Lemma \ref{genReltrclassHydro} that
$(\ppp U\ppp)^{1/2}(f_{\gamma,\kappa}+1)^{-s}\in \gS^2(\ghpk).$

Eventually, we verify
\begin{equation}
    \label{eq:NZ}
  (\Pi_\kappa^+ U \Pi_\kappa^+)^{2s}
  \leq a_{\gamma,s,\kappa}(f_{\gamma,\kappa}+1)^{2s'}
  \end{equation}
  for $1/2<s'<s$.  To show \eqref{eq:NZ}, we use \eqref{eq:SD} with
  $f(x)=x^{2s}$ and obtain
\begin{align*}
  (\Pi^+_\kappa U \ppp)^{2s}
  \leq \Pi_\kappa^+ U^{2s} \Pi_\kappa^+.
\end{align*}
Hence, by \eqref{eq:genSobolevDualHydro}, the left side of
\eqref{eq:NZ} is bounded by
$(f_{\gamma,\kappa}+1)^{2s'}$ times a constant,
since $U^{2s}\in\cK_{s'}^{(0)}$.
\end{proof}

\subsection{Controlling Coulomb perturbations}
The main difficulty in applying Proposition \ref{diffgen0} is
verifying \eqref{eq:relbounddelta}. In our setting it is an inequality
for each fixed $\kappa\in \zp$. However, it follows from the
following stronger statement which does not need a partial wave
analysis.
\begin{lemma}
  \label{boundpertcoulomb}
  Let $\gamma\in(0,1)$, $0\leq U\in r^{-1}L^\infty([0,\infty))$, and
  $\frac12\leq s<\min\{\frac32-\sigma_\gamma,1\}$.
  Then there is a $a_{\gamma,s}\in\rz$ and a $\lambda_0>0$ such that
  for $|\lambda|<\lambda_0$
  \begin{equation}
    (F_\gamma+1)^{2s}
    \leq a_{\gamma,s} (F_{\gamma}(\lambda U)+1)^{2s}. 
  \end{equation}
\end{lemma}
\begin{proof}
  Since $F_\gamma>-1$, obviously, for sufficiently small $\lambda$,
  $F_{\gamma}+1 -\lambda \Lambda_\gamma U \Lambda_\gamma>0$.

  Next we first assume $\lambda>0$.  By operator convexity of
  $x\mapsto x^{2s}$ with $s\in[1/2,1]$ and the Davis-Sherman
  inequality \eqref{eq:SD}, we obtain
  \begin{multline}
    \label{eq:convexity}
          \left(F_\gamma+1\right)^{2s}
      = \left(F_{\gamma}(\lambda U)+1 + \lambda \Lambda_\gamma U\Lambda_\gamma \right)^{2s}
      \leq 2^{2s-1}\left(F_{\gamma}(\lambda U)+1 \right)^{2s}\\
      + 2^{2s-1}\lambda^{2s}\left(\Lambda_\gamma U\Lambda_\gamma \right)^{2s}
      \leq  2^{2s-1}\left(F_{\gamma}(\lambda U)+1 \right)^{2s} +
      2^{2s-1}\lambda^{2s}\Lambda_\gamma U^{2s}\Lambda_\gamma.
      \end{multline}
  Since $U(r)\leq \|rU\|_\infty/r$, Hardy's inequality yields
  $ U^2 \leq 4\|rU\|_\infty^2|p|^2$.  Thus, by operator monotonicity
  of roots and Corollary \ref{hardydomcor},
  \begin{align*}
    \Lambda_\gamma U^{2s}\Lambda_\gamma \leq 4^s\|rU\|_\infty^{2s}\Lambda_\gamma |p|^{2s}\Lambda_\gamma
    \leq 4^s\|rU\|_\infty^{2s}A_{s,\gamma}\left(F_\gamma+1\right)^{2s}
  \end{align*}
  where $A_{s,\gamma}$ is the constant in \eqref{eq:hardydomcorform}.
  Plugging this estimate in \eqref{eq:convexity} yields
  $$
  \left(F_\gamma+1\right)^{2s} \leq
  2^{2s-1}(1-2^{4s-1}A_{s,\gamma}\|rU\|_\infty^{2s}\lambda^{2s})^{-1}
  \left(F_{\gamma}(\lambda U)+1\right)^{2s}
  $$
  proving the assertion for $\lambda>0$.
  
  If $\lambda<0$, we set $\mu:=-\lambda>0$ and
  $\varepsilon:=\sqrt\mu\in(0,1)$.  By operator convexity
  \begin{align}
    \label{eq:convexity2}
       (F_\gamma+1)^{2s}
       = &\left[(1-\varepsilon)(F_\gamma(-\mu U)+1 + \varepsilon\left(F_\gamma((1-\varepsilon)\varepsilon^{-1}\mu U)+1 \right)\right]^{2s}\\
 \label{40a}     \leq &2^{2s-1}(1-\varepsilon)^{2s}(F_\gamma(-\mu U)+1 )^{2s}\\
    \label{40b}    & + 2^{2s-1}\varepsilon^{2s}(F_\gamma((1-\varepsilon)\varepsilon^{-1}\mu U)+1)^{2s}.
  \end{align}
  Here, we used that both operators are non-negative (because of the
  condition on $\lambda=-\mu$ which implies that that the coupling
  constant of the perturbation in the second summand is
  $\mathcal{O}(\sqrt\mu)$ which is chosen sufficiently small). This
  allows to use operator convexity in \eqref{40a}. Suppose there is
  an $\upsilon\in\rz$ and a $\mu_0>0$ such that for all $\mu\in[0,\mu_0]$
  \begin{equation}
    \label{40c}
  \eqref{40b} \leq \varepsilon^{2s}\upsilon(F_\gamma+1)^{2s},
  \end{equation}
  then, the assertion follows as in the case $\lambda>0$ by taking this
  term to the left side of \eqref{eq:convexity2} and dividing both sides
  of the inequality by $(1-\varepsilon^{2s}\upsilon)$, which is allowed
  for sufficiently small $\varepsilon$.

  We turn to the proof of \eqref{40c}: since
  $\Lambda_\gamma D_\gamma=\Lambda_\gamma|D_\gamma|$ and by
  \eqref{eq:SD}, we have
  $$
  \left(\Lambda_\gamma\left(D_\gamma-\frac{1-\varepsilon}{\varepsilon}\mu U\right)\Lambda_\gamma\right)^{2s} \leq \Lambda_\gamma\left(|D_\gamma|-\frac{1-\varepsilon}{\varepsilon}\mu U\right)^{2s}\Lambda_\gamma.
  $$
  Now, $(\balpha\cdot\bp+\beta)^2=p^2+1$,
  $U(r)\leq\tfrac1r\|rU\|_\infty $, Schwarz' and Hardy's
  inequality imply
  \begin{align*}
    \left(|D_\gamma|-\frac{1-\varepsilon}{\varepsilon}\mu U\right)^2
    & \leq 4\left(|p|^2+1+\left(\gamma+\frac{1-\varepsilon}{\varepsilon}\mu\|rU\|_\infty\right)^2|x|^{-2}\right)\\
    & \leq 4\left(\left(1+4\left(\gamma+\frac{1-\varepsilon}{\varepsilon}\mu\|rU\|_\infty\right)^2\right)|p|^2+1\right)\\
    & \leq 4\left(1+4\left(\gamma+\frac{1-\varepsilon}{\varepsilon}\mu\|rU\|_\infty\right)^2\right)(|p|+1)^2.
  \end{align*}
  Using $(|p|+1)^{2s}\leq2^{2s-1}(|p|^{2s}+1)$ for $s\in[1/2,1]$, and
  applying Corollary \ref{hardydomcor}, we see the existence of the
  wanted $\upsilon$, if $\mu$, and thus $\mu/\epsilon$, is
  sufficiently close to zero.
\end{proof}

\begin{proof}[Proof of Proposition \ref{diffu1}]
  We apply Proposition \ref{diffgen0} to the operators
  $ A = f_{\gamma,\kappa}$ and $B =\ppp U \ppp\geq 0$ with
  $1/2<s<\min\{3/2-\sigma_\gamma,1\}$.

We have already verified the assumptions concerning $f_{\gamma,\kappa}$ in
the proof of Proposition \ref{diffu2}.  The fact that
$\ppp U\ppp$ is relatively form bounded with
respect to $f_{\gamma,\kappa}$ follows from Kato's inequality and Corollary
\ref{hardydomcor} in every channel $\kappa$.

Since $U\in\cK_{s}^{(0)}$, Lemma \ref{genReltrclassHydro} implies
$$
\tr(f_{\gamma,\kappa}+1)^{-s}\ppp U\ppp(f_{\gamma,\kappa}+1)^{-s}<\infty.
$$

Finally, Assumption \eqref{eq:relbounddelta} follows from
$$
(F_\gamma+1)^{2s}\leq A_{\gamma,s}(F_{\gamma}(\lambda U)+1)^{2s}
$$ for
small $|\lambda|$ which is Lemma \ref{boundpertcoulomb}.
\end{proof}

\section{Controlling large angular momenta}
\label{s:alll}

\subsection{Estimating the spectral shift in channel $\kappa$}
We will use the notation $D_\gamma(\varphi)$ and $F_\gamma(\varphi)$
introduced in \eqref{eq:dphi} and \eqref{eq:fphi}. The following
proposition will allow to apply the Weierstra{\ss} M-test to deduce
Theorem \ref{convalll} from Theorem \ref{convfixedl} as in
\cite{Franketal2019P}.

\begin{proposition}
  \label{dominantsigma}
  Let $0<\gamma<1$, $0\leq V(r)\leq\gamma/r$ for $r\in\rz_+$, and
  $U=U_1+U_2$ with $0\leq U_1\in r^{-1}L^\infty_c (\rz_+)$ and
  $0\leq U_2\in\cD$.
  Then there are $\varepsilon>0$, $A_{\gamma,s}<\infty$, $\lambda_1>0$,
  and $K_\gamma\in\nz$ such that $|\lambda|<\lambda_1$ and
  $|\kappa|\geq K_\gamma$ implies
  \begin{equation}
    \tr_\kappa F_0(V+\lambda U)_- -\tr_\kappa F_0(V)_- 
   \leq A_{\gamma,s}\lambda \|U\|_{\cK_{s,0}} |\kappa|^{-1-\varepsilon}.
  \end{equation}
\end{proposition}

In preparation of the proof, we give a trace and a Sobolev inequality
with respect to $C_{\ell_\kappa}+a \kappa^{-2}$ on $\gh_\kappa$. We recall that
these inequalities were crucial in \cite{Franketal2019P} to treat
functions belonging to $\cK_{s,\delta}$.  The following lemma follows
from \cite[Proposition 22]{Franketal2019P} in the same way as Lemma
\ref{genReltrclassnod} followed from \cite[Proposition
19]{Franketal2019P}.

\begin{lemma}
  \label{genReltrclass}
  Let $a>0$, $\delta\in[0,2s-1]$, $s\in(1/2,3/4]$, $\kappa\in\zp$, and
  $0\leq W \in\cK_{s,\delta}$. Then,
  \begin{align}
    \label{eq:genReltrclass}
    \|W^{1/2}(C_{\ell_\kappa}+a\kappa^{-2})^{-s}\|_{\gS^2(\gh_{\kappa})}^2
    \leq A_{s,a} |\kappa|^{1-\delta}\|W\|_{\cK_{s,\delta}}. 
  \end{align}
  In particular, in $L^2(\rz^3:\cz^4)$,
  \begin{align}
    \label{eq:genSobolevDual}
    \Pi_{\kappa}W\Pi_{\kappa}
    \leq A_{s,a} |\kappa|^{-\delta}\|W\|_{\cK_{s,\delta}} (\Pi_{\kappa}(\sqrt{-\Delta+1}-1+a\kappa^{-2})\Pi_{\kappa})^{2s}.
  \end{align}
\end{lemma}

Note, that the latter inequality follows immediately, since the
multiplicity of each eigenvalue is proportional to $|\kappa|$.

To prove Proposition \ref{dominantsigma}, we will again control
Dirac operators by scalar operators:
\begin{lemma}
  \label{boundkineticenergy}
  Let $a>0$ and $\kappa\in\zp$ such that $1\geq a \kappa^{-2}$. Then
  $$
  (D_0-1+a \kappa^{-2})^2 \geq (\sqrt{p^2+1}-1+a \kappa^{-2})^2.
  $$
\end{lemma}

\begin{proof}
  The assertion is equivalent to the inequality
  $$
    p^2+1+(1-\tfrac a{\kappa^2})^2-2(1-\tfrac a{\kappa^2})D_0\\
  \geq p^2+1+(1-\tfrac a{\kappa^2})^2-2(1-\tfrac a{\kappa^2})\sqrt{p^2+1}.
  $$
  Since $1\geq a \kappa^{-2}$ and
  $D_0=\balpha\cdot\bp+\beta \leq
  |\balpha\cdot\bp+\beta|=\sqrt{p^2+1}$, the assertion follows.
\end{proof}

\begin{lemma}
  \label{boundA}
  Let $\gamma\in(0,1)$, $0\leq V(r)\leq\gamma/r$, $s\in(1/2,3/4]$, and
  $U=U_1+U_2$ with $U_1\in r^{-1}L^\infty([0,\infty))$ and
  $|U_2|^{2s}\in \cK_{s,0}$.  Then there are constants
  $K_{\gamma}\in\nz$ and $a_\gamma,\lambda_2\in \rz_+$, such that for
  all $\lambda\in[-\lambda_2, \lambda_2]$ and all $\kappa\in\zp$ with
  $|\kappa|\geq K_\gamma$
  \begin{equation}
    \label{agamma}
    f_{0,\kappa}(V+\lambda U) \geq -a_\gamma \kappa^{-2}
  \end{equation}
  holds.
\end{lemma}
\begin{proof}
  Note that Sommerfeld's eigenvalue formula \eqref{eq:eigenvalue}
  immediately implies \eqref{agamma} for pure Coulomb potentials
  $V(r)+\lambda U(r)=\gamma/r$. In this case it will be useful to
  emphasize the Coulombic origin and write $a_{\gamma,C}$ instead of
  $a_\gamma$.

  Since
  $ \int_0^\infty(\ell+\frac12)^2r^{-2}|g(r)|^2\,\dr \leq
  (g,(p_\ell^{(r)})^2g)_{L^2(\rz_+)}$ (Hardy) and by picking
  $\ell=\ell_\kappa=|\kappa|-\theta(\kappa)$ (see \eqref{l}), we
  have, initially for $f\in\gh_{\kappa,m}$ but extending to
  $f\in\gh_\kappa$,
  \begin{multline}
    \label{eq:hardyang}
    ( f,|x|^{-2}f)_{L^2(\rz^3:\cz^4)}\\
    \leq {(
      f^+,(p_{|\kappa|-\theta(\kappa)}^{(r)})^2f^+)_{L^2(\rz_+)}
      \over (|\kappa|-\frac12\sgn(\kappa))^2} + {(
      f^-,(p_{|\kappa|-\theta(-\kappa)}^{(r)})^2f^-)_{L^2(\rz_+)}\over
      (|\kappa|+\frac12\sgn(\kappa))^2} \leq 2{(
      f,p_{\ell_\kappa}^2 f)_{L^2(\rz^3:\cz^4)}\over
      \kappa^2}.
   \end{multline}

  Since there exist $d\in\rz_+$ such that for all $b\in\rz_+$ and $\ell\in\nz_0$
   \begin{align*}
    \|(p_\ell^{(r)}+b)(C_\ell^{(r)}+b)^{-1}\|_{L^2(\rz_+,\dr)}
    \leq d\left(b^{-1/2}\one_{\{b\leq1\}} + \one_{\{b>1\}}\right)
  \end{align*}
  (Frank et al \cite[Formula (48)]{Franketal2019P}), this implies with
  $b:=a_\gamma\kappa^{-2}$ and $\ell:=\ell_\kappa$
  \begin{equation}
    \label{eq:plancherelang}
    \begin{split}
      & \|(p_{\ell_\kappa}+a_\gamma \kappa^{-2})(C_{\ell_\kappa}+a_\gamma
      \kappa^{-2})^{-1}\|_{\gh_{\kappa}} \leq
      d\left(\frac{|\kappa|}{\sqrt{a_\gamma}}\one_{\{a_\gamma\leq
          \kappa^2\}} + \one_{\{a_\gamma>\kappa^2\}}\right).
    \end{split}
  \end{equation}
  We claim that $K_\gamma:=\lceil \sqrt{a_\gamma}\rceil+1$ and
  $a_\gamma:=\max\{a_{\gamma,C},2d^2\}$ are constants that have the
  claimed properties: the triangle inequality, Lemma
  \ref{boundkineticenergy}, and the estimates \eqref{eq:hardyang} and
  \eqref{eq:plancherelang} imply for $f\in\ghpk$
  \begin{equation}
    \label{eq:dcc}
    \begin{split}
    &\|(F_\gamma+\tfrac{a_\gamma}{\kappa^2})f\| =
    \|(D_0-1-\tfrac\gamma{|x|}+\tfrac{a_\gamma}{\kappa^2})f\|
    \geq \|(D_0-1+\tfrac{a_\gamma}{\kappa^2})f\|-\gamma\||x|^{-1}f\|\\
    \geq &(1-\gamma\sqrt{2/a_\gamma}d)
    \|(C_{\ell_\kappa}+\tfrac{a_\gamma}{\kappa^2})f\| \geq (1-\gamma)
    \|(C_{\ell_\kappa}+\tfrac{a_\gamma}{\kappa^2}) f\|,
    \end{split}
    \end{equation}
    by definition of $a_\gamma$ and $K_\gamma$.
     Thus,
    $$\left(\left.C_{\ell_\kappa}\right|_\ghpk+a_\gamma \kappa^{-2})\right)^2
    \leq (1-\gamma)^{-2}(f_{\gamma,\kappa}+a_\gamma \kappa^{-2})^2.
    $$
  By operator monotonicity of the square root and since
  $f_{\gamma,\kappa}+a_\gamma \kappa^{-2}\geq0$ 
  the last bound implies
  \begin{equation}
    \label{eq:boundChandraDirac}
       \left.C_{\ell_\kappa}\right|_\ghpk+\tfrac{a_\gamma }{\kappa^2}\leq (1-\gamma)^{-1}(f_{\gamma,\kappa}+\tfrac{a_\gamma}{\kappa^2})\leq (1-\gamma)^{-1}(f_{0,\kappa}(V)+\tfrac{a_\gamma}{\kappa^2}). 
  \end{equation}
  Next, \eqref{eq:hardyang} and \eqref{eq:plancherelang} allow us to
  estimate
  \begin{equation}
   U_1 \leq \|rU_1\|_\infty (C_{\ell_\kappa}+a_\gamma \kappa^{-2}).
  \end{equation}
  Moreover, by \eqref{eq:genSobolevDual} and the definition of $\cK_{s,0}$ 
  $$
  |U_2|^{2s}
  \leq A_{s,a_\gamma}\||U_2|^{2s}\|_{\cK_{s,0}}(C_{\ell_\kappa} +a_\gamma \kappa^{-2})^{2s}
   $$
   Thus, by operator monotonicity of $x\mapsto x^s$ with $s\in(0,1]$,
  $$
  U \leq
  \left[\|rU_1\|_\infty+A_{s,a_\gamma}^{1/(2s)}\||U_2|^{2s}\|_{\cK_{s,0}}^{1/(2s)}\right]
  (C_{\ell_\kappa}+a_\gamma \kappa^{-2}).
  $$
  Combining this bound with \eqref{eq:boundChandraDirac}, we obtain
  for sufficiently small $|\lambda|$,
  $$
    f_{0,\kappa}(V)+a_\gamma \kappa^{-2}
    \geq (1-\gamma)(C_{\ell_\kappa}+a_\gamma \kappa^{-2})
    \geq \lambda U,
  $$
  thereby proving the assertion.
\end{proof}

We are now ready to prove Proposition \ref{dominantsigma}.

\begin{proof}[Proof of Proposition \ref{dominantsigma}]
  Let $d_{\kappa,\lambda}$ denote the orthogonal projection onto the
  negative spectral subspace of $F_0(V+\lambda U)$ in $\ghpk$. By the
  variational principle, we obtain
  \begin{equation} 
    s_{\kappa,\lambda}  := \tr_\kappa F_0(V+\lambda U)_-
    - \tr_\kappa F_0(V)_-
    \leq \lambda\tr(d_{\kappa,\lambda}U)
    \end{equation}
  Similar to \cite[Equation (19)]{Iantchenkoetal1996} we set
  \begin{equation}
    \begin{split}
      A :=&d_{\kappa,\lambda}\left(f_{0,\kappa}(V+\lambda U)+b_\kappa\right)^s,\
      B  :=\left((f_{0,\kappa}(V+\lambda U)+b_\kappa)\right)^{-s}\Lambda_\gamma(C_{\ell_\kappa}+b_\kappa)^s,\\
      C :=&(C_{\ell_\kappa}+b_\kappa)^{-s} U
      (C_{\ell_\kappa}+b_\kappa)^{-s}
    \end{split}
  \end{equation}
  yielding
  $$
  s_{\kappa,\lambda}\leq\lambda \tr (ABCB^*A^*).
  $$
  We choose $ b_\kappa:=b_\gamma/\kappa^2$ with some sufficiently
  large $b_\gamma$, that is going to be determined later.  We start by
  estimating $\|A\|$ using Lemma \ref{boundA} which is applicable
  since $U_2^{2s}\in\cK_{s',4(s-s')}\subseteq\cK_{s,0}$ by Lemma
  \ref{inclusions}. Since $d_{\kappa,\lambda}$ projects onto the
  negative spectral subspace of $F_0(V+\lambda U)$ on $\ghpk$, Lemma
  \ref{boundA} implies that there are $\lambda_2>0$ and
  $K_\gamma\in\nz$ such that for all $\lambda\in\rz$ with
  $|\lambda|<\lambda_2$ and all $\kappa\in\zp$ with
  $|\kappa|\geq K_\gamma$, we have
  $f_{0,\kappa}(V+\lambda U)+b_\gamma\kappa^{-2}\geq(b_\gamma-a_\gamma)
  \kappa^{-2}$ which is  strictly positive for $b_\gamma>a_\gamma$
  which we will assume from now on. In particular
  $\|A\|\leq b_\gamma^{s} |\kappa|^{-2s}$.
  
  Next,
  $\|C\|_{\gS^1(\gh_{\kappa})} \leq A_{s,b_\gamma}|\kappa|\|U\|_{\cK_{s,0}}$
  is an immediate consequence of Lemma \ref{genReltrclass}.

  We now show the boundedness of $B$. We write $B=B_1B_2$ where 
  \begin{align*}
    B_1 & := \left(f_{0,\kappa}(V+\lambda U)+b_\gamma\kappa^{-2})\right)^{-s} \left(f_{0,\kappa}(V+\lambda U_1)+b_\gamma\kappa^{-2})\right)^s\\
    B_2 & := \left(f_{0,\kappa}(V+\lambda U_1)+b_\gamma\kappa^{-2}\right)^{-s}\Lambda_\gamma(C_{\ell_\kappa}+b_\gamma\kappa^{-2})^s
  \end{align*}
  as operators in $\ghpk$.  To estimate $\|B_2\|$, we wish to show
  \begin{equation}
    \label{eq:boundB2squared}
    \ppp(C_{\ell_\kappa}+b_\gamma\kappa^{-2})^{2}\ppp
    \leq 4\left((f_{0,\kappa}(V+\lambda U_1)+b_\gamma\kappa^{-2})\right)^{2}.
  \end{equation}
  Believing this estimate for the moment, we can use the operator
  monotonicity of $x\mapsto x^{s}$ with $s\in(0,1]$ and the following
  inequality by Frank and Geisinger \cite[Lemma
  6.4]{FrankGeisinger2016} which is closely related to the
  Davis-Sherman inequality \eqref{eq:SD}.  Namely, if $T\geq0$ is a
  linear operator with trivial kernel, $P$ an orthogonal projection,
  and $f$ an operator monotone function on $\rz_+$, then
  $$
  P f(T)P \leq Pf(PTP)P.
  $$
  As in the discussion after \eqref{eq:SD}, the right side simplifies
  to $f(PTP)$, if $f(0)=0$.  In our case,
  $T=(C_{\ell_\kappa}+b_\gamma\kappa^{-2})^2$ in $\ghpk$,
  $P=\Lambda_\gamma$, $f(x)=x^s$, and $0<s\leq1$, i.e., Frank and
  Geisinger's inequality reads in $\ghpk$
  \begin{equation}
    \label{eq:SDgen}
       \Lambda_\gamma (C_{\ell_\kappa}+{b_\gamma\over\kappa^2})^{2s}\Lambda_\gamma
       \leq \Lambda_\gamma\left(\Lambda_\gamma(C_{\ell_\kappa}+{b_\gamma\over\kappa^2})^2\Lambda_\gamma\right)^s\Lambda_\gamma
      = \left(\Lambda_\gamma(C_{\ell_\kappa}+{b_\gamma\over\kappa^2})^2\Lambda_\gamma\right)^s.
       \end{equation}
  Combining this inequality with \eqref{eq:boundB2squared} would
  establish the boundedness of $B_2$.
  
  To prove \eqref{eq:boundB2squared}, we use the triangle inequality
  and $\Lambda_\gamma\leq1$ and estimate for $f\in\ghpk$
  \begin{multline}
    \label{eq:boundU1quad1}
    \|(F_0(V+\lambda U_1)+{b_\gamma\over\kappa^2})f\|=   \|\Lambda_\gamma(D_0-\gamma/r+\gamma/r-V-\lambda U_1-1+{b_\gamma\over\kappa^2})\Lambda_\gamma f\|\\
      \geq \|\Lambda_\gamma(D_\gamma-1+b_\gamma\kappa^{-2})\Lambda_\gamma f\| - (\gamma+|\lambda|\|rU_1\|_\infty)\||x|^{-1}\Lambda_\gamma f\|.
  \end{multline}
  Using that $\Lambda_\gamma$ and $D_\gamma-1+b_\gamma\kappa^{-2}$ commute and
  Lemma \ref{boundkineticenergy}, we estimate further
  \begin{multline}
    \label{eq:boundU1quad2}
   \|\Lambda_\gamma(D_\gamma-1+b_\gamma\kappa^{-2})\Lambda_\gamma f\|
      =\|(D_\gamma-1+b_\gamma\kappa^{-2})\Lambda_\gamma f\|\\
      \geq \|(D_0-1+b_\gamma\kappa^{-2})\Lambda_\gamma f\| - \gamma\||x|^{-1}\Lambda_\gamma f\|
      \geq \|(C_{\ell_\kappa}+b_\gamma\kappa^{-2})\Lambda_\gamma f\| - \gamma\||x|^{-1}\Lambda_\gamma f\|.
    \end{multline}
  Combining \eqref{eq:boundU1quad1} and \eqref{eq:boundU1quad2} with
  \eqref{eq:hardyang} and \eqref{eq:plancherelang}, we obtain
  \begin{align*}
    & \|\Lambda_\gamma(D_0(V+\lambda U_1)-1+b_\gamma\kappa^{-2})\Lambda_\gamma f\|\\
    & \quad \geq \left[1-\sqrt{2}d(2\gamma+\lambda\|rU_1\|_\infty)\left(b_\gamma^{-\frac12}\one_{\{b_\gamma\leq \kappa^2\}} + |\kappa|^{-1}\one_{\{b_\gamma\geq \kappa^2\}}\right)\right] \|(C_{\ell_\kappa}+b_\gamma\kappa^{-2})\Lambda_\gamma f\|.
  \end{align*}
  Choosing
  \begin{equation}
    \label{eq:gammacond}
    b_\gamma = 2\max\{a_\gamma,8d^2(2\gamma+\lambda_2\|rU_1\|_\infty)^2\}
    \quad\text{and}\quad
    K_\gamma = \left\lceil\sqrt{b_\gamma}\right\rceil
  \end{equation}
  with $a_\gamma$ as in Lemma \ref{boundA} shows
  \begin{equation}
    \label{eq:boundU1quad}
    (f_{0,\kappa}(V+\lambda U_1)+b_\gamma\kappa^{-2})^2
    \geq \frac14 \Lambda_\gamma(C_{\ell_\kappa}+b_\gamma\kappa^{-2})^2|_\ghpk
  \end{equation}
  for all $\nz\ni|\kappa|\geq K_\gamma$ and $|\lambda|<\lambda_2$,
  thereby establishing \eqref{eq:boundB2squared}.
  Using \eqref{eq:boundB2squared}, operator monotonicity of $x\mapsto x^s$
  for $s\in(0,1]$, and \eqref{eq:SDgen}, we eventually obtain
  \begin{align}
    \label{eq:boundB2}
    \Lambda_\gamma(C_{\ell_\kappa}+b_\gamma\kappa^{-2})^{2s}|_\ghpk
    \leq 4^s(f_{0,\kappa}(V+\lambda U_1)+b_\gamma\kappa^{-2})^{2s}
      \end{align}
  for all $|\kappa|\geq K_\gamma$ and $|\lambda|<\lambda_2$.
  This shows $\|B_2\|<2^s$.
   
  Now, we turn to $B_1$ and show 
  \begin{equation}
    \label{eq:boundB1}
       (f_{0,\kappa}(V+\lambda U_1)+b_\gamma\kappa^{-2})^{2s}\\
      \leq 2(f_{0,\kappa}(V+\lambda U_1)+b_\gamma\kappa^{-2}-\lambda \ppp U_2\ppp)^{2s}.
    \end{equation}
     By \cite[Lemma 15]{Franketal2019P}, which we recall in Lemma
  \ref{apriori}, estimate \eqref{eq:boundB1} holds, provided we can show
  \begin{equation}
    \label{eq:Neidhardt}
       \|\left|\lambda \Lambda_\gamma U_2\Lambda_\gamma\right|^{s}\left(D_0(V+\lambda U_1)|_{\ghpk}-1+{b_\gamma\over2\kappa^2}\right)^{-s'}\|
     \leq A_{s,s'}\left({b_\gamma\over2\kappa^2}\right)^{s-s'}
   \end{equation}
  for a certain constant $A_{s,s'}$ and some $1/2<s'<s$.
  To show \eqref{eq:Neidhardt}, we first use the Davis-Sherman inequality
  \eqref{eq:SD} and \eqref{eq:genSobolevDual} and obtain
  \begin{align*}
    &|\lambda \Lambda_\gamma U_2\Lambda_\gamma|^{2s}
     \leq |\lambda|^{2s} \Lambda_\gamma U_2^{2s}\Lambda_\gamma\\
    \leq &4^{s'}A_{s',b_\gamma} |\lambda|^{2s} |\kappa|^{-4(s-s')}\|U_2^{2s}\|_{\cK_{s',4(s-s')}}\Lambda_\gamma(C_{\ell_\kappa}+b_\gamma\kappa^{-2})^{2s'}\Lambda_\gamma
      \quad\text{in}\ \gh_{\kappa}.
  \end{align*}
  Combining this estimate with \eqref{eq:boundB2} with $s$ replaced
  by $s'$, i.e.,
  \begin{align*}
    \|(C_{\ell_\kappa}+b_\gamma\kappa^{-2})^{s'}\Lambda_\gamma\left(D_0(V+\lambda U_1)|_\ghpk-1+b_\gamma\kappa^{-2}/2\right)^{-s'}\| \leq 2^{s'}
  \end{align*}
  shows that the left side of \eqref{eq:Neidhardt} is bounded by
  $4^{s'}A_{s',b_\gamma}^{1/2}|\lambda|^s\|U_2^{2s}\|_{\cK_{s',4(s-s')}}^{1/2}
  |\kappa|^{-2(s-s')}$.
 
  Thus, there is a $\lambda_3>0$ such that \eqref{eq:Neidhardt} holds
  for all $|\lambda|<\lambda_3$ which shows $\|B\|\leq A_{s}$, uniformly
  in $\lambda$ and $\kappa$. 
    
  Combining the bounds on $\|A\|^2$, $\|B\|^2$, and $\|C\|_1$, we find
  for $|\lambda|<\min\{\lambda_2,\lambda_3\}$ and all $|\kappa|\geq K_\gamma$,
  $$    s_{\kappa,\lambda} 
  \leq A_{\gamma,s}\lambda \|U\|_{\cK_{s,0}} |\kappa|^{1-4s}$$
  what was claimed since $s>1/2$.
\end{proof}

\subsection{Proof of Theorem \ref{existencerhoh} on the existence of $\rho^H$}
We will now prove the pointwise bounds on $\rho_{\kappa}^H$ of Theorem
\ref{existencerhoh}.  The strategy of the proof is similar to the one
of Proposition \ref{dominantsigma}.

Let $d_\kappa$ denote the orthogonal projection onto the negative
spectral subspace of $F_\gamma$ in $\ghpk$.  Then
$$
\rho_{\kappa}^H(x) = \tr d_\kappa \delta^{(s)}_{|x|} = \tr ABCB^*A^*
$$
where $\delta^{(s)}_R$ is the delta sphere function with radius $R$,
i.e., $\delta^{(s)}_R(y):=\delta(|y|-R)/(4\pi R^2)$ and
\begin{equation}
  \begin{split}
  A:=&d_\kappa(f_{\gamma,\kappa}+\tilde a_\kappa)^s,\
  B:=(f_{\gamma,\kappa}+\tilde a_\kappa)^{-s}\Lambda_\gamma(C_{\ell_\kappa}+\tilde a_\kappa)^s,\\\
  C:=&(C_{\ell_\kappa}+\tilde a_\kappa)^{-s}\ \delta_{|x|}^{(s)} (C_{\ell_\kappa}+\tilde a_\kappa)^{-s}
  \end{split}
\end{equation}
with $\tilde a_\kappa:=a_{\gamma,C}\kappa^{-2}$, $a_{\gamma,C}$ as
defined in the beginning of the proof of Lemma \ref{boundA}.
Moreover, the parameter $s$ obeys $1/2<s<3/2-\sigma_\gamma$ and
$s\leq3/4$.

First, we have $\|A\|^2\leq a_{s,\gamma} |\kappa|^{-4s}$
by \eqref{eq:eigenvalue}.

Next,
$\tr C={4|\kappa|\over|x|^2}[(C_{\ell_\kappa}^{(r)}+\tilde
a_\kappa)^{-2s}(|x|,|x|) + (C_{|\kappa|-\theta(-\kappa)}^{(r)}+\tilde
a_\kappa)^{-2s}(|x|,|x|)]$.  Here, it is crucial to have $s>\frac12$,
since $\delta_r$ on $\rz_+$ is not form bounded with respect to
$(C_{\ell_\kappa}^{(r)})^{2s}$ for any $s\leq\frac12$.  The diagonal
was estimated in \cite[Lemma 26]{Franketal2019P} for
$s\in(\tfrac12,\tfrac34]$, namely
\begin{multline*}
  (C_{\ell_\kappa}^{(r)}+\tilde a_\kappa)^{-2s}(r,r) \\
  \leq A_{s,a_\gamma}\left[\left(\frac{r}{|\kappa|}\right)^{2s-1}\one_{\{r\leq |\kappa|\}} + \left(\frac{r}{|\kappa|}\right)^{4s-1}\one_{\{|\kappa| \leq r \leq \kappa^2\}}
  + |\kappa|^{4s-1}\one_{\{r \geq \kappa^2\}}\right].
\end{multline*}
Repeating this computation for $(C_{|\kappa|-\theta(-\kappa)}^{(r)}+\tilde a_\kappa)^{-2s}(r,r)$
shows that the same bound holds also in this case, since
$|\ell_\kappa-(|\kappa|-\theta(-\kappa))|=1$.

The uniform boundedness of $B$ in $|\kappa|$ was
shown in the proof of Proposition \ref{dominantsigma} for $\gamma<1$
and all $|\kappa|\geq K_\gamma$ where $K_\gamma$ is given in
\eqref{eq:gammacond}.

For $|\kappa|\leq K_\gamma$, the uniformity of estimates on $\|B\|$
with respect to $|\kappa|$ is not crucial, since only a fixed finite
number of angular momentum channels is involved.  For these $\kappa$,
we write
\begin{equation}
  B=(B_1\circ B_2\circ B_3)^*
\end{equation}
where
\begin{align*}
  B_1  := &(C_{\ell_\kappa}+\tilde a_\kappa)^s (p_\ell+\tilde a_\kappa)^{-s}, \
  B_2 := (p_\ell+\tilde a_\kappa)^{s}((f_{\gamma,\kappa}+1+\tilde a_\kappa))^{-s},\\
  B_3 := &(f_{\gamma,\kappa}+1+\tilde a_\kappa)^{s}((f_{\gamma,\kappa}+\tilde a_\kappa))^{-s}.
\end{align*}
Clearly, $\|B_1\|\leq A_{s,\gamma}$ in each channel.  By Corollary
\ref{hardydomcor}, respectively \eqref{eq:hardydomcorang}, for fixed
$\gamma<1$ and $s<\frac32-\sigma_\gamma$ if
$\gamma\geq\frac{\sqrt3}2$, we have $\|B_2\|\leq A_{s,\gamma}$ for all
$\kappa\in\zp$.  By \eqref{eq:eigenvalue}
\begin{equation}
   \|(f_{\gamma,\kappa}+\tilde a_\kappa+1)(f_{\gamma,\kappa}+\tilde a_\kappa)^{-1}\| \leq 1 + A \kappa^2 \leq A_{K_\gamma}.
\end{equation}
Thus, by operator monotonicity of $x\mapsto x^s$ with $s\in(0,1]$ and
the bounds on $B_1,B_2$, and $B_3$, we obtain
$ \|B\|\leq A_{s,\gamma}$. Combining the bounds on $A,B$, and $C$, we
obtain
\begin{multline*}
  \rho_{\kappa}^H(x)  \leq A_{s,\gamma} {|\kappa|^{1-4s}\over|x|^2}\\
  \times \left[\left(\frac{|x|}{|\kappa|}\right)^{2s-1}\one_{\{|x| \leq |\kappa|\}}+\left(\frac{|x|}{|\kappa|}\right)^{4s-1}\one_{\{|\kappa| \leq |x| \leq |\kappa|^2\}}
                       + |\kappa|^{4s-1}\one_{\{|x| \geq |\kappa|^2\}}\right].
\end{multline*}
In particular, the right side is summable for $s>1/2$ and one finally
obtains
\begin{align*}
  \rho^H(x)
  = \sum_{\kappa\in\zp}^\infty \rho_\kappa^H(x)
  \leq A_{s,\gamma} (|x|^{2s-3}\one_{\{|x|\leq1\}} + |x|^{-3/2}\one_{\{|x|>1\}}).
\end{align*}
Recalling the assumptions on $s$ concludes the proof of Theorem
\ref{existencerhoh}.\qed

\section{Proof of the strong Scott conjecture}
We are now in position to prove Theorem \ref{convfixedl}, i.e., the
strong Scott conjecture for fixed angular momentum, or to be more
accurate, for fixed spin-orbit coupling $\kappa$.

Since the statement of Theorem \ref{convfixedl} is linear with respect
to $U$, we can assume without loss of generality that $U$ is
non-negative and belongs either to
$r^{-1}L^\infty_c (\rz_+)$ or to $\cD_\gamma^{(0)}$.

Given a spherically symmetric potential $U$ define $U_c$ by
$U_c(x):=c^2U(cx)$. Furthermore, using
\eqref{eq:manyfurry} for $N=Z$ and fixing $\kappa\in\zp$, we introduce
the quadratic form
\begin{equation}
    \cE_{c,Z,\lambda,\kappa}:\bigwedge_{n=1}^Z\Lambda_{c,Z}(\cS(\rz^3:\cz^4))\rightarrow\cz,\
    \psi\mapsto \cE_{c,Z}[\psi]-\lambda \sum_{\nu=1}^Z(\psi, (\Pi_\kappa\circ U_c\circ\Pi_\kappa)_\nu\psi),
\end{equation}
if this form is defined and bounded from below.

If $U\in\cD_\gamma^{(0)}$ then $U^{2s}\in\cK_{s'}^{(0)}$ and thus
$(\Pi^+_{\kappa}U\Pi_{\kappa}^+)^{2s} \leq
a_{s,s',\gamma}(f_{\gamma,\kappa}+1)^{2s'} $ for $1/2<s'<s$ by the
proof of Proposition \ref{diffu2}. Thus, by Proposition \ref{diffgen}
and Lemma \ref{apriori}, $U$ is infinitesimally form bounded with
respect to $f_{\gamma,\kappa}$.  Hence, $\cE_{c,Z,\lambda,\kappa}$ is
defined and bounded from below.

If $U\in r^{-1}L^\infty_c ([0,\infty))$, the same follows
from Kato's inequality and Corollary \ref{hardydomcor} for all
$\lambda$ in an $Z$ independent open neighborhood of zero.

Obviously, we can rewrite the expectation of the one-particle
perturbation $\Pi^+_\kappa U_c\Pi^+_\kappa$ in the state $\Lambda$ in
terms of its ground state density (see \eqref{eq:furryangdensity})
$\rho_{\kappa}(x)$ in channel $\kappa$; we have
\begin{equation}
  \label{eq:4.2}
  \int_{\rz^3} \rho_\kappa(x)U_c(x)\,\dx
  = {1\over D\lambda}\sum_{d=1}^D(\cE_{c,Z}[\psi_d]-\cE_{c,Z,\lambda,\kappa}[\psi_d]).
\end{equation}
It obviously depends only superficially on $\lambda$.  To estimate
this from above we pick $\lambda>0$, use the upper bound on
$\cE_{c,Z}$ \cite{HandrekSiedentop2015} (Scott correction for the
Furry operator) and a lower bound on $\cE_{c,Z,\lambda,\kappa}$ by the
correlation inequality of Mancas et al \cite{Mancasetal2004}
(MMS). This reduces the problem to a one-particle Furry operator with
screened Coulomb potential given by the Thomas-Fermi density. This
one-particle problem can be treated by the methods developed in the
previous sections. The corresponding lower will be for free by
reversing the sign of $\lambda$.

We begin with the lower bound on $\cE_{c,Z,\lambda,\kappa}$ recalling a
special case of the MMS correlation inequality: we write
$\rho_Z^{\mathrm{TF}}$ for the minimizer of the Thomas-Fermi functional
for a neutral atom (Lieb and Simon \cite[Theorem
II.20]{LiebSimon1977}). Next we define a ball centered at $x$ with
radius $R_Z(x)$ defined by
$\int_{|x-y| \leq R_Z(x)}\rho_Z^{\mathrm{TF}}(y)\,\dy=\frac12$.  The
screening potential to be used is
$\chi_Z(x):=\int_{|x-y| \geq
  R_Z(x)}\rho_Z^{\mathrm{TF}}(y)|x-y|^{-1}\,\dy.$ Then,
\cite{Mancasetal2004}
\begin{equation}
  \label{MMS}
  \sum_{1\leq\nu<\mu\leq Z}|x_\nu-x_\mu|^{-1}
  \geq \sum_{\nu=1}^Z\chi_Z(x_\nu)-D[\rho_Z^{\mathrm{TF}}].
\end{equation}
This allows to eliminate all two-particle terms in $\cE_{c,Z,\lambda,\kappa}$.
\begin{lemma}
  \label{correlation}
  For sufficiently small $\lambda$ and all $L\in\nz$
  \begin{multline}
    \label{Korrelation}
    D^{-1}\sum_{d=1}^D\cE_{c,Z,\lambda,\kappa}[\psi_d]  \\
    \geq -\sum_{|{\kappa'}|<L}\tr_{\kappa'}(F_{c,Z}(\lambda\Pi_{\kappa}
    U_c\Pi_{\kappa})_- - \sum_{Z/2\geq |{\kappa'}|\geq L}\tr_{\kappa'} F_{c,Z}(-\chi_Z+\lambda \Pi_\kappa
    U_c\Pi_\kappa)_- - D[\rho_Z^{\mathrm{TF}}]
  \end{multline}
\end{lemma}

\begin{proof}
  By MMS, using the one-particle density matrix $\gamma_\Lambda$ of
  $\Lambda$, partial wave analysis and $\chi_Z\geq0$, we have
  \begin{multline}
    D^{-1}\sum_{d=1}^D\cE_{c,Z,\lambda,\kappa}[\psi_d]  \\
    \geq \sum_{|{\kappa'}|<L}\tr_{\kappa'}[\gamma_\Lambda
    F_{c,Z}(\lambda\Pi_\kappa U_c\Pi_\kappa)] + \sum_{|{\kappa'}|\geq
      L}\tr_{\kappa'}[\gamma_\Lambda F_{c,Z}(-\chi_Z+\lambda \Pi_\kappa
    U_c\Pi_\kappa)] - D[\rho_Z^{\mathrm{TF}}]
  \end{multline}
  for any $L\in\nz$. Since the energy is increasing in $|\kappa'|$
  (see Lemma \ref{boundA}) and $\tr\gamma_\Lambda=Z$ we can minimize in
  the one-particle density matrix under this constraint. The resulting
  summands will surely vanish, if $|{\kappa'}|>Z$, i.e., we can cut
  off the series at $Z$. Dropping the requirement of $Z$ particles
  only gives the wanted result.
\end{proof}

Next we recall ground state energy is
$$
D^{-1}\sum_{d=1}^D\cE_{c,Z}[\psi_d] = E^{\mathrm{TF}}(Z) + \left(\frac12-s(\gamma)\right)Z^2+\mathcal{O}(Z^{47/24})
\quad \text{as}\ Z\to\infty
$$
(Handrek and Siedentop \cite[Theorem 1]{HandrekSiedentop2015}) with the finite spectral shift
\begin{equation}
  \label{eq:spectralshift}
  s(\gamma)
  =\gamma^{-2}\sum\nolimits_{{\kappa'}\in\zp}[\tr_{\kappa'} (F_\gamma)_- - |{\kappa'}|\sum\nolimits_{n\in\nz}\gamma^2(n+\ell_{\kappa'})^{-2}].
\end{equation}
In fact, with $L:=Z^{1/9}$ their proof gives the stronger chain of inequalities 
\begin{equation}
  \label{scottoperator}
  \begin{split}
  &-\const Z^{47/24} +
  E^{\mathrm{TF}}(Z)+\left(\tfrac12-s(\gamma)\right)Z^2\\
  \leq &-\sum_{|{\kappa'}|<L}\tr_{\kappa'} (F_{c,Z})_- - \sum_{Z/2\geq
    |{\kappa'}|\geq L}\tr_{\kappa'} F_{c,Z}(-\chi_Z)_- -
  D[\rho_Z^{\mathrm{TF}}]\leq D^{-1}\sum_{d=1}^D\cE_{c,Z}[\psi_d] \\
  \leq &E^{\mathrm{TF}}(Z)+\left(\tfrac12-s(\gamma)\right)Z^2+\const Z^{47/24}
  \end{split}
  \end{equation}
  implying
  \begin{multline}
    \label{wichtig}
    D^{-1}\sum_{d=1}^D\cE_{c,Z}[\psi_d]\\
    = -\sum_{|{\kappa'}|<Z^\frac19}\tr_{\kappa'} (F_{c,Z})_- - \sum_{Z/2\geq
      |{\kappa'}|\geq Z^\frac19}\tr_{\kappa'} F_{c,Z}(-\chi_Z)_- -
    D[\rho_Z^{\mathrm{TF}}] + O(Z^\frac{47}{24}).
  \end{multline}

  \begin{proof}[Proof of Theorem \ref{convfixedl}]
    We divide \eqref{eq:4.2} by $c^2$, use the bounds \eqref{wichtig},
    rescale the bounds $x\to x/c$, and use \eqref{Korrelation}. We get
    for positive $\lambda$
    \begin{multline}
      \label{eq:reduced}
      \int_{\rz^3}c^{-3}\rho_\kappa(x/c)U(x)\rd x
      =\frac1{c^2D}\sum_{d=1}^D(\cE_{c,Z}[\psi_d]-\cE_{c,Z,\lambda,\kappa}[\psi_d])\\
      \leq \frac1\lambda
      \begin{cases}
        \tr[f_{\gamma,\kappa}(\lambda U)_--f_{\gamma,\kappa}(0)_-] &|\kappa|<Z^\frac19\\
        \tr[f_{\gamma,\kappa}(-c^{-2}\chi_Z(\cdot/c)+\lambda U)-f_{\gamma,\kappa}(-c^{-2}\chi_Z(\cdot/c))_-]&
        |\kappa|\geq Z^\frac19)
      \end{cases}
      + \const Z^{-\frac1{24}}
    \end{multline}
    where we use that $c=\gamma^{-1}Z$. Taking $c$ to $\infty$ gives
    \begin{equation}
      \label{oben}
      \limsup_{c\to\infty}\int_{\rz^3}c^{-3}\rho_\kappa(x/c)U(x)\rd x
      \leq {\tr f_{\gamma,\kappa}(\lambda U)_--\tr f_{\gamma,\kappa}(0)_-\over\lambda}. 
    \end{equation}
      Taking $\lambda<0$ gives the reverse inequality
      \begin{equation}
      \label{unten}
      \liminf_{c\to\infty}\int_{\rz^3}c^{-3}\rho_\kappa(x/c)U(x)\rd x
      \geq {\tr f_{\gamma,\kappa}(\lambda U)_--\tr f_{\gamma,\kappa}(0)_-\over\lambda}. 
    \end{equation}
    By Propositions \ref{diffu2} and \ref{diffu1}, the right sides of
    \eqref{oben} and \eqref{unten} tend to
    $\int_{\rz^3} \rho_{\kappa}^H(x)U(|x|)\,\dx$ thus yielding the
    existence of the limit and its limit quod erat demonstrandum.
\end{proof}

\begin{proof}[Proof of Theorem \ref{convalll}] The proof follows from
  the previous one by summing over $\kappa$ provided we can
  interchange the summation and the limits $Z\to\infty$ and
  $\lambda\to0$. This, however is secured by Proposition
  \ref{dominantsigma} which allows to apply the Weierstra\ss\
  criterion (Lebesgue dominated converge).
\end{proof}
Note that given Proposition \ref{dominantsigma} the proof of Theorem
\ref{convalll} is analogous to \cite[Theorem 2]{Iantchenkoetal1996}
(see also Frank et al \cite[Theorem 2]{Franketal2019P}).

\appendix

\section{Partial Wave Analysis\label{Anhang1}}
\label{a:partialwaves}
We collect some notations and known facts about the partial wave
analysis of Dirac operators (see, e.g., Evans et al
\cite{Evansetal1996}, Balinsky and Evans \cite[Section
2.1]{BalinskyEvans2011}, and Thaller \cite[Sections
4.6.3-4.6.5]{Thaller1992}).

Let $Y_{\ell,m}$ be the spherical harmonics on the unit sphere
$\bbS^2$ obeying the normalization condition
$\int_{\bbS^2}|Y_{\ell,m}|^2\,\domega=1$ where $\domega$ is the usual
surface measure on $\bbS^2$.  If $|m|>\ell$, we set
$Y_{\ell,m}\equiv0$.  We begin by observing that those of the
spherical spinors
\begin{equation}\label{2.6} \Omega_{\ell,m,s}(\omega) :=
  \begin{pmatrix} 2s\sqrt{\frac{\ell+\frac12+2sm}{2\ell+1}}
    Y_{\ell,m-\frac12}(\omega)\\ \sqrt{\frac{\ell+\frac12-2sm}{2\ell+1}}
    Y_{\ell,m+\frac12}(\omega)
  \end{pmatrix} 
\end{equation}
with $\ell=0,1,2,...$ and $m=-\ell-\frac12,...,\ell+\frac12$, that do not
vanish, form an orthonormal basis of $L^2(\mathbb{S}^2:\cz^2)$
(see, e.g., Evans et al \cite[Equation (7)]{Evansetal1996}).

Moreover, they are joint eigenfunctions of $L^2$, $J^2$ ($J=L+S$ being
the total angular momentum), and $J_3$ with respective eigenvalues
$\ell(\ell+1)$, $(\ell+s)(\ell+s+1)$, and $m$.

Introducing the spin-orbit operator $K=\beta(J^2-L^2+1/4)$, there is
an orthonormal basis of eigenvectors $\Phi^\sigma_{\kappa,m}$ of
$L^2(\bbS^2:\cz^4)$ such that
$J^2\Phi^\sigma_{\kappa,m}=j_\kappa(j_\kappa+1)\Phi^\sigma_{\kappa,m}$,
$J_3\Phi^\sigma_{\kappa,m}=m\Phi^\sigma_{\kappa,m}$, and
$K\Phi^\sigma_{\kappa,m}=\kappa\Phi^\sigma_{\kappa,m}$ with
$j_\kappa:=|\kappa|-1/2$ introduced in \eqref{l},
$m\in\{-j_\kappa,...,j_\kappa\}$, $\kappa\in\zp$, and
$\sigma\in\{+,-\}$. A standard choice is
\begin{equation}\label{fi}
  \Phi^+_{\kappa,m} :=  \left(\begin{array}{c}
          \ri\sgn(\kappa)\Omega_{\ell_\kappa,m,\frac12\sgn(\kappa)}\\
          0
                            \end{array}
                          \right), \
                          \Phi^-_{\kappa,m}:=\left(\begin{array}{c}
                              0\\
                              -\sgn(\kappa)\Omega_{\ell_\kappa+\sgn(\kappa),m,-\frac12\sgn(\kappa)}
                            \end{array}
                          \right).    
\end{equation}
Using these spinors, we introduce
\begin{align}
  \label{a3}
  \gh_{\kappa,m} :=& \mathrm{span}\{x\mapsto \tfrac{f^+(|x|)}{|x|} \Phi_{\kappa,m}^+(\tfrac x{|x|}) + \tfrac{f^-(|x|)}{|x|}\Phi_{\kappa,m}^-(\tfrac x{|x|}):f^+,f^-\in L^2(\rz_+)\},\\
     \label{a2}
  \gh_\kappa:=& \bigoplus_{m=-j_\kappa}^{j_\kappa}\gh_{\kappa,m},\ \gh_\kappa^+:=\Lambda_\gamma\gh_\kappa  
\end{align}
These spaces form an orthogonal decomposition of $L^2(\rz^3:\cz^4)$ and
$\Lambda_\gamma(L^2(\rz^3:\cz^4))$.

We write $\Pi_\kappa$, $\Pi_{\kappa,m}$, and $\Pi_\kappa^+$ for the
orthogonal projection onto $\gh_{\kappa}$, $\gh_{\kappa,m}$, and
$\ghpk$. If we write -- in abuse of notation --
$\Phi_{\kappa,m}^\pm(\omega,\tau)$ for the $\tau$-th component of
$\Phi_{\kappa,m}^\pm(\omega)$, we can write the action of $\Pi_\kappa$ on
$g\in L^2(\rz^3\otimes\{1,...,4\})$ more explicitly as
\begin{equation}
  \label{a1}
  (\Pi_\kappa g)(r\omega,\tau)
  = \sum_{\sigma\in\{+,-\}}\sum_{m=-j_\kappa}^{j_\kappa} \Phi^\sigma_{\kappa,m}(\omega,\tau)
    \sum_{\tau'=1}^4\int_{\mathbb{S}^2}\rd \omega'\, \overline{ \Phi^\sigma_{\kappa,m}(\omega',\tau')}g(r\omega',\tau')
\end{equation}
writing $x=r\omega$ with $r:=|x|$ and $\omega:=x/r$.

Note that Dirac operators with spherical potentials leave the space
$\gh_{\kappa,m}$ invariant which can be seen explicitly in
\eqref{eq:radialdirac}.  Moreover, their eigenvalues depend on
$\kappa$ only.

Furthermore note, that $\Pi^+_\kappa$ is also an orthogonal
projection, since $\Pi_\kappa$ commutes with $\Lambda_\gamma$ (see
\cite[Equation (27)]{HandrekSiedentop2015}).

\section{Test function spaces}
\label{s:testfunctions}
The test functions for which we prove the strong Scott conjecture
belong to the function spaces $\cK_{s}^{(0)}$ and $\cK_{s,\delta}$ which
were already introduced in Frank et al \cite{Franketal2019P} and are
defined as
\begin{align}
  \label{eq:defsigmadeltanod}
  \begin{split}
    \cK_{s}^{(0)} & :=\{W\in L^1_{\mathrm{loc}}(\rz_+):\|W\|_{\cK_{s}^{(0)}}<\infty\}\\
    \|W\|_{\cK_{s}^{(0)}} & := \int_0^1 r^{2s-1}|W(r)|\,\dr + \int_1^\infty |W(r)|\,\dr
  \end{split}
\end{align}
and
\begin{align}
  \label{eq:defsigmadelta}
  \begin{split}
    \cK_{s,\delta} & :=\{W\in L^1_{\mathrm{loc}}(\rz_+):\|W\|_{\cK_{s,\delta}}<\infty\}\\
    \|W\|_{\cK_{s,\delta}} & :=\sup_{R\geq1} R^\delta\left[\int_0^R \left(\frac{r}{R}\right)^{2s-1}|W(r)|\,\dr + \int_R^{R^2}\left(\frac{r}{R}\right)^{4s-1} |W(r)|\,\dr\right.\\
    &\qquad\qquad\qquad \left.+ R^{4s-1}\int_{R^2}^\infty |W(r)|\,\dr\right]
  \end{split}
\end{align}
for $s\geq1/2$ and $\delta\in[0,2s-1]$.
Here, $L^p_{\mathrm{loc}}(\rz_+)$ denotes the space of all functions that
belong to $L^p$ on any compact subset of $\rz_+$.
We note some basic inclusion properties which already occurred
implicitly in \cite{Franketal2019P}.
\begin{lemma}
  \label{inclusions}
  Let $1/2\leq s'<s$ and $\delta\in[0,2s-1]$.
  Then the spaces $\cK_{s}^{(0)}$ and $\cK_{s,\delta}$
  obey the following inclusion properties.
  \begin{enumerate}
  \item One has $\cK_{s'}^{(0)}\subseteq\cK_{s}^{(0)}$.
    
  \item One has $\cK_{s,\delta}\subseteq\cK_{s,0}\subseteq\cK_{s}^{(0)}$.

  \item One has $\cK_{s',4(s-s')} \subseteq \cK_{s,0}$, if additionally
    $1/2<2s/3+1/6 \leq s'<s$.
  \end{enumerate}
\end{lemma}
This means that functions must be smoother at the origin the smaller
$s$ is. Moreover, functions belonging to $\cK_{s,\delta}$ must decay
faster at infinity than those belonging to $\cK_{s}^{(0)}$.

To give a digestible representation of our convergence results,
we introduce the test function spaces
\begin{align}
  \label{eq:deftestnod0}
  \begin{split}
    \cD_\gamma^{(0)} =
    \begin{cases}
      & \{ W\in\cK_{s}^{(0)}:\  |W|^{2s}\in\cK_{s'}^{(0)} \ \text{for some}\ 1/2<s'<s\leq 1 \} \\
      & \qquad\qquad\qquad \text{if}\ 0<\gamma<\sqrt3/2 .
      \\
      & \{ W\in\cK_{s}^{(0)}:\  |W|^{2s}\in\cK_{s'}^{(0)}  \ \text{for some}\ 1/2<s'<s<3/2-\sigma_\gamma \} \\
      & \qquad\qquad\qquad \text{if}\ \sqrt3/2\leq\gamma<1 .
  \end{cases}
  \end{split}
\end{align}
and
\begin{align}
  \label{eq:deftest0}
  \begin{split}
    \cD & = \{W \in\cK_{s,0}\,:\  |W|^{2s}\in\cK_{s',4(s-s')}\ \text{for some}\ 1/2<2s/3+1/6\leq s'<s\leq 3/4 \}.
  \end{split}
\end{align}
We refer to \cite{Franketal2019P} for an alternative and more convenient
representation of the space $\cD_\gamma^{(0)}$ as well as the norm
$\|W\|_{\cK_{s,\delta}}$ (see their Formulae (34) and (43)).
For instance, $r^{-1}L^\infty_c$ functions belong both to
$\cK_{s}^{(0)}$ and $\cK_{s,\delta}$ for $s>1/2$ and $\delta\in[0,2s-1]$.
Moreover, one easily verifies
$L^1(\rz_+)\subseteq \cK_{s}^{(0)}$ and
$L^1(\rz_+,r^\delta\,\dr)\cap L^1(\rz_+,r^{4s-1+\delta}\,\dr)\subseteq\cK_{s,\delta}$
for $s\geq1/2$ and $\delta\in[0,2s-1]$.

\section{Auxiliary tools}

The following lemma, which we quote from \cite[Lemma 15]{Franketal2019P}
was inspired by Neidhardt and Zagrebnov \cite[Lemma 2.2]{NeidhardtZagrebnov1999}.

\begin{lemma}
  \label{apriori}
  Let $A$ be a self-adjoint operator with $\inf\sigma(A)>0$ and let $B$ be an
  operator which satisfies $B\geq 0$ or $B\leq 0$. Assume that for some
  numbers $\max\{s',1/2\}<s<1$ one has
  \begin{align*}
    \||B|^s A^{-s'}\|<\infty .
  \end{align*}
  Then $B$ is form bounded with respect to $A$ with relative bound zero and,
  if $\||B|^s A^{-s'}\| \leq A_{s,s'}M^{s-s'}$ for some constant $A_{s,s'}$
  depending only on $s$ and $s'$,
  $$
  \frac12 (A+M)^{2s} \leq (A+B+M)^{2s} \leq 2(A+M)^{2s} .
  $$
\end{lemma}

\section*{Acknowledgments}

This research was partly carried out at the Institute for Mathematical
Sciences at the National University of Singapore during the program
\textit{Density Functionals for Many-Particle Systems: Mathematical
  Theory and Physical Applications of Effective Equations}.  We are
grateful to the IMS and the Julian Schwinger foundation for their
hospitality and financial support. Special thanks go to Berthold-Georg
Englert who was the heart of the program.

Partial financial support by the Deutsche Forschungsgemeinschaft
(DFG, German Research Foundation) through grant SI 348/15-1 (H.S.) and
through Germany's Excellence Strategy – EXC-2111 – 390814868 (H.S.)
is gratefully acknowledged.

\def\cprime{$'$}

\end{document}